\documentclass[journal]{IEEEtran}
\IEEEoverridecommandlockouts
\usepackage[noadjust]{cite} % get rid of auto-spacing
\usepackage{amsmath, amssymb, amsfonts}
\usepackage{algorithmic}
\usepackage{graphicx, subcaption}
\usepackage{textcomp}
\usepackage{xcolor}
\def\BibTeX{{\rm B\kern-.05em{\sc i\kern-.025em b}\kern-.08em
    T\kern-.1667em\lower.7ex\hbox{E}\kern-.125emX}}

\usepackage{mathtools}
\input{mysymbol.sty}
\usepackage{cleveref} %% for clever referencing
\usepackage{siunitx}

\usepackage[colorinlistoftodos]{todonotes}

\colorlet{LightTeal}{white!70!teal}

\definecolor{bostonuniversityred}{rgb}{0.8, 0.0, 0.0}

\newlength\myindent % define a new length \myindent
\usepackage{amsthm}
\newtheorem{assumption}{\hspace{0pt}\bf Assumption}
\newtheorem{lemma}{\hspace{0pt}\bf Lemma}
\newtheorem{proposition}{\hspace{0pt}\bf Proposition}

\newtheorem{theorem}{\hspace{0pt}\bf Theorem}
\newtheorem{corollary}{\hspace{0pt}\bf Corollary}

\newtheorem{remark}{\hspace{0pt}\bf Remark}

\newtheorem{definition}{\hspace{0pt}\bf Definition}

\usepackage{algorithmic}
\usepackage{algorithm}
\usepackage{listings}
\usepackage{enumerate}
\usepackage{enumitem}
\usepackage{bbm}

\def\x{{\mathbf x}}

\date{\today}

\def\E{\mathbb{E}}

\newcommand{\norm}[1]{\left\|#1\right\|}
\usepackage{caption}
\begin{document}

% \title{A State-Augmented Primal-Dual Algorithm for \\ Near-Optimal Radio Resource Management \\
% % {\footnotesize \textsuperscript{*}Note: Sub-titles are not captured in Xplore and
% % should not be used}

\title{Fast State-Augmented Learning for Wireless \\Resource Allocation with Dual Variable Regression}

\author{Yi\u{g}it~Berkay~Uslu, Navid~NaderiAlizadeh, Mark~Eisen,~\IEEEmembership{Member,~IEEE,}
~Alejandro~Ribeiro,~\IEEEmembership{Fellow,~IEEE}% <-this % stops a space

\thanks{
Yi\u{g}it Berkay Uslu and Alejandro Ribeiro are with the Department of Electrical
and Systems Engineering, University of Pennsylvania, Philadelphia, PA 19104 USA (e-mail: ybuslu@seas.upenn.edu; aribeiro@seas.upenn.edu).
}
\thanks{
Navid NaderiAlizadeh is with the Department of Biostatistics and Bioinformatics, Duke University, Durham, NC 27705 USA (e-mail: navid.naderi@duke.edu).
}
\thanks{
Mark Eisen is with Johns Hopkins Applied Research Laboratory (APL), Laurel, MD 11100 USA (e-mail: mark.eisen@ieee.org).
}
}

\maketitle
\thispagestyle{plain}
\pagestyle{plain}

\begin{abstract}
We consider resource allocation problems in multi-user wireless networks, where the goal is to optimize a network-wide utility function subject to constraints on the ergodic average performance of users. We demonstrate how a state-augmented graph neural network (GNN) parametrization for the resource allocation policy circumvents the drawbacks of the ubiquitous dual subgradient methods by representing the network configurations (or states) as graphs and viewing dual variables as dynamic inputs to the model, treated as graph signals supported over the graphs. Lagrangian maximizing state-augmented policies are learned during the offline training phase, and the dual variables evolve through gradient updates while executing the learned state-augmented policies during the inference phase. Our main contributions are to illustrate how near-optimal initialization of dual multipliers for faster inference can be accomplished with dual variable regression, leveraging a secondary GNN parametrization, and how maximization of the Lagrangian over the multipliers sampled from the dual descent dynamics substantially improves the training of state-augmented models. We demonstrate the superior performance of the proposed algorithm with extensive numerical experiments in a case study of transmit power control. Finally, we prove a convergence result and an exponential probability bound on the excursions of the dual function (iterate) optimality gaps.
\end{abstract}

\section{Introduction}
\label{sec:intro}

The resource allocation literature spans a broad range of problem formulations and algorithmic approaches. To name a few directions, classical optimization and information-theoretic methods address interference management through scheduling and spectrum sharing~\cite{naderializadeh2014itlinq, yi2015itlinq+, shen2017fplinq}. For distributed optimization, primal-domain methods such as Laplacian-gradient flows and their accelerated variants have been proposed in~\cite{doostmohammadian2024accelerated, doostmohammadian2025momentum} (see~\cite{doostmohammadian2025survey} for a comprehensive survey). Dual decomposition and alternating direction method of multipliers (ADMM)-based approaches~\cite{shi2011iteratively, luo2008dynamic} tackle similar problem classes by operating in the dual domain. On the learning side, a growing body of work~\cite{eisen2020optimal, naderializadeh2021resource, niknam2020intelligent, chowdhury2021unfolding, wang2021unsupervised, letaief2021edge, zhao2021link, nikoloska2021modular, li2022power, huang2023regularization, uslu2024learning, das2025opportunistic, darabi2025diffusion, garcia2025linkscheduling, uslu2025generative, zhao2025backpressure} has demonstrated strong performance by leveraging neural network parametrizations to learn resource allocation policies directly from data.

A general class of stochastic, nonconvex resource optimization problems is considered in \cite{eisen2019learning, eisen2020optimal, ribeiro2012optimal, naderializadeh2022learning}, where the goal is to maximize a network-wide utility function subject to constraints on the long-term average performance of users. For such problems, parametrizing the resource allocation policies and working in the Lagrangian dual domain is preferred. Alternating between Lagrangian maximization over the primal variables and dual subgradient updates yields a \emph{dual gradient descent (DGD) algorithm} that provably samples from a near-optimal and almost-surely feasible stochastic policy \cite{ribeiro2010ergodic}. However, this algorithm faces two practical challenges. Namely, maximizing the Lagrangian at each dual iteration is costly to perform online, and the algorithm samples from an optimal policy only asymptotically. Recently, a class of \emph{state-augmented algorithms} \cite{calvo2021state, StateAugmented_RRM_GNN_naderializadeh_TSP2022, das2025opportunistic} have been proposed to alleviate these challenges by \emph{training a single policy parametrization} that takes both network states and dual variables as inputs and learns Lagrangian maximizers for all dual variable values jointly \emph{during offline training}. At inference, only lightweight dual updates are needed, which produce the temporal randomization of the learned policies needed to sample from a stochastic resource allocation policy with ergodic near-optimality and feasibility guarantees.

In this paper, we propose a \emph{novel state-augmented primal-dual learning algorithm with dual variable regression (SA+DR)} for stochastic optimization problems in wireless networks. Our method improves upon prior implementations of state augmentation~\cite{StateAugmented_RRM_GNN_naderializadeh_TSP2022, calvo2021state} in two important ways.
First, we introduce a dual-regression (DR) policy that initializes the dual multipliers close to their optimal levels during online execution, replacing the default zero initialization used in prior work. This substantially shortens the transient period during which online dual dynamics execute suboptimal and potentially infeasible policies. 
The dual-regressor is trained after the state-augmented policy (SA) by optimizing a regression loss, where the regression targets are near-optimal dual multipliers estimated by rolling out state-augmented dual descent dynamics over the training network configurations, and the regression inputs are the corresponding network configurations. Second, we introduce a training subroutine that samples dual multipliers from dual descent roll-outs during state-augmented training. This replaces the fixed uniform prior used in earlier works with dual multiplier samples that better reflect the distribution encountered during inference, eliminating the need for manual tuning of the training distribution and yielding state-augmented policies with superior performance.

Aside from the algorithmic novelties, we derive a convergence result showing that the dual function iterates regularly visit a small neighborhood of the dual optimum, and an excursion result showing that large deviations between visits are exponentially unlikely. These results characterize the suboptimality of arbitrary multiplier initializations and provide theoretical justification for the dual-regression initialization.

We demonstrate the SA+DR algorithm in a power control setup resembling that of~\cite{StateAugmented_RRM_GNN_naderializadeh_TSP2022}, where the objective is to maximize the ergodic sum-rate, subject to per-user minimum-rate requirements. Both the resource allocation policy and the dual variable regression policy are parametrized by graph neural network architectures, referred to as primal-GNN and dual-GNN, respectively, that operate on a graph representation of the wireless network configuration. In this representation, users correspond to nodes, channel gains define weighted directed edges, and resource allocations, dual variables, and user utilities are treated as node signals on the graph.

We validate the proposed SA+DR method through extensive ablations against state-augmented algorithms without optimal initialization and other baselines, demonstrating faster convergence to near-optimal and feasible policies. We illustrate the policy switching behavior through which the state-augmented algorithm samples from a stochastic policy that is near-optimal and feasible in the ergodic sense, and verify that both the primal-GNN and dual-GNN policies generalize.

The paper is organized as follows. Section~\ref{sec:formulation} describes the problem formulation. In Section~\ref{sec:alg}, we introduce our proposed state-augmented algorithm with dual variable regression. Section~\ref{sec:convergence_results} presents the convergence and excursion results. We introduce the GNN-parametrizations, and present numerical experiments in Section~\ref{sec:experiments}. We conclude the paper in Section~\ref{sec:conclusion} and defer additional details to the Appendices. 
\section{Wireless Resource Allocation Problem}\label{sec:formulation}

This paper focuses on allocating network resources (e.g., frequency bands, transmit powers, etc.) optimally to maximize a network-wide utility function (e.g., sum-rate) subject to constraints on each user's performance (e.g., minimum rate requirements). Both the objective and constraints are defined in terms of expectations of nonconvex functions, and the constraints must hold for each network configuration.

\subsection{Network Setup and Optimization Problem}

Consider a wireless network comprised of $N$ users and operating over a series of time steps $t\in\{0,1,2,\dots,T-1\}$. The channel model includes both large-scale and small-scale (fast) fading effects. Accordingly, we decompose the network state given by the set of all channel gains at time step $t$ into two components. The large-scale channel gains, collected in the network configuration (state) $\bbH \in \ccalH$, are determined primarily by path-loss and shadowing, and they remain constant throughout the $T$ steps. The small-scale channel gains fluctuate over time due to fast-fading effects, and together with the large-scale component, produce the instantaneous network states $\bbH_t \in \ccalH$ at time step $t$. We assume that the allocation policies have access only to the large-scale network configuration.

For a given $\bbH$, we define the resource allocation function (policy) $\bbp: \ccalH \mapsto \reals^N$, and $\bbp(\bbH)$ denotes the vector of resource allocation decisions. These decisions lead to the network-wide ergodic performance vector $\bbr \big(\bbH, \bbp(\bbH) \big) \in \reals^N$, with $\bbr: \ccalH \times \reals^N \mapsto \reals^N$ denoting the performance function. The ergodic performance vector is defined through the instantaneous performance vector,
\begin{equation}
\begin{aligned} \label{eq:ergodic-performance-function}
    \bbr \big( \bbH, \bbp(\bbH) \big) \!&\coloneqq\! \E_{\h_t \vert \h} \left[ \bbr \big( \bbH_t, \bbp(\bbH) \big) \right] \!\approx\! \frac{1}{T} \sum_{t = 0}^{T-1} \bbr \big( \bbH_t, \bbp(\bbH) \big),
\end{aligned}
\end{equation}
where, with a slight abuse of notation, we denote by $\bbr$ both ergodic and instantaneous performance functions.

Given a concave network-wide objective utility $f_0: \reals^N \to \reals$, a concave constraint utility $\bbf : \reals^N \to \reals^N$, and any $\h$, we write the unparametrized resource allocation problem as
\begin{subequations}\label{eq:nonparam_problem}
\begin{alignat}{2}
    P^\star(\h) = &\maximum_{\bbp(\h)} &~&  f_0 \left( \bbr \big(\bbH, \bbp(\bbH) \big) \right),\label{eq:objective_non_param}  \\
    &~\st &~&  \bbf \left( \bbr \big(\bbH, \bbp(\bbH) \big) \right) \geq \bbzero.
    % ,\quad  \Dh\text{-a.e.}%
\end{alignat}
\end{subequations}
\noindent In~\eqref{eq:nonparam_problem}, we wrote $\f = (f_1, \ldots, f_N)^\top$ to denote the vector-valued function whose $i$th entry $f_i$ corresponds to the constraint utility of the $i$th agent. Note that the objective and the constraints are defined in terms of the \emph{ergodic average} network performance vector, rather than the instantaneous performance vector [cf.~\eqref{eq:ergodic-performance-function}]. The goal is to find the optimal allocations $\bbp^\star(\bbH)$ that maximize the objective utility $f_0$ while satisfying all constraint utilities $\bbf$ for any given $\h$. 

\begin{remark}
Although we instantiate $\bbH$ in terms of large-scale channel gains, our formulation, algorithms, and the theoretical analysis carry through unchanged if $\bbH$ instead represents noisy, quantized, or otherwise imperfect  measurements.
\end{remark}

\subsection{Parametrized Dual Problem}
Solving~\eqref{eq:nonparam_problem} entails an infinite-dimensional search over all resource allocation functions $\bbp$ for every $\bbH$, which is practically infeasible. We therefore restrict our search to \emph{parametrized} resource allocation functions $\bbp_{\bbtheta}(\bbH; \bbtheta)$ indexed by a finite-dimensional set of parameters $\bbtheta \in \bbTheta$, and consider a family (distribution) of network configurations $\Dh$. This yields the parametrized resource allocation problem,
\begin{subequations}\label{eq:param_problem}
\begin{alignat}{2}
    P_{\bbtheta}^{\star}=& \maximum_{\bbtheta \in \bbTheta} ~&& \E_{\h \sim \Dh} \left[ f_0 \Big( \bbr \big(\bbH, \bbp_{\bbtheta}(\bbH; \bbtheta) \big) \Big) \right] ,\label{eq:objective_non_param}  \\
    &~\st &~&  \bbf \Big( \bbr \big(\bbH, \bbp_{\bbtheta}(\bbH; \bbtheta) \big) \Big) \geq \bbzero,\quad \Dh\text{-a.e.} \label{eq:min_rate_constraint_param}
\end{alignat}
\end{subequations}
where, different from~\eqref{eq:nonparam_problem}, the maximization is now performed over the set of parameters $\bbtheta \in \bbTheta$ and the utility constraints are to hold for almost all network configurations $\h \sim \Dh$. The problem defined in \eqref{eq:param_problem} is not convex in $\bbtheta$ in general, and obtaining approximate maximizers with feasibility guarantees is challenging. We work around this by moving to the Lagrangian dual domain. Introducing the nonnegative dual multipliers $\bblambda: \ccalH \to \reals^N_{+}$ for the constraints of each $\bbH$ in~\eqref{eq:min_rate_constraint_param}, we define the parametrized Lagrangian for \eqref{eq:param_problem} as
\begin{align}
  \ccalL_{\bbtheta}(\bbtheta, \bblambda) & \coloneqq \E_{\h \sim \Dh} \bigg[ \ccalL_{\bbtheta}\big( \bbtheta, \bblambda(\bbH); \bbH \big)  \bigg], \label{eq:param_lagrangian}
\end{align}
where $\ccalL_{\bbtheta}(\bbtheta, \bblambda(\bbH); \bbH)$ is defined as
\begin{align} \label{eq:param_lagrangian_decomposed}
    \ccalL_{\bbtheta}(\bbtheta, \bblambda(\bbH); \bbH) &\coloneqq f_0 \Big( \bbr \big(\bbH, \bbp_{\bbtheta}(\bbH; \bbtheta) \big) \Big) \nonumber \\
  &\qquad+ \bblambda(\bbH) \cdot \bbf \Big( \bbr \big(\bbH, \bbp_{\bbtheta}(\bbH; \bbtheta) \big) \Big).
\end{align}
Since the constraints hold per-configuration and the objective sums over configurations, the problem decomposes into independent subproblems indexed by $\bbH$. We exploit this separable structure and suppress the dependence on $\bbH$ throughout.

We define a Lagrangian maximizer,
\begin{align}\label{eq:lagrangian_maximizer_theta}
\bbtheta^{\dagger}(\bblambda) \in \argmax_{\bbtheta \in \bbTheta} \Ltheta(\bbtheta, \bblambda),
\end{align}
and write the parametrized dual function as 
\begin{align} \label{eq:parametrized-dual-function}
\gtheta(\bblambda) \coloneqq \Ltheta \big( \bbtheta^\dagger(\bblambda), \bblambda \big) = \max_{\bbtheta \in \bbTheta} \ccalL_{\bbtheta}(\bbtheta, \bblambda).
\end{align} 
Since the dual function $g_{\bbtheta}(\bblambda)$ upper bounds $P^\star_{\bbtheta}$ for any $\bblambda \succcurlyeq \bb0$, the dual problem seeks the tightest upper bound,
\begin{align}\label{eq:param_dual_problem}
    D_{\bbtheta}^{\star} \coloneqq g_{\bbtheta}(\bblambda^\star) = \min_{\bblambda} g_{\bbtheta}(\bblambda).
\end{align}
\noindent In~\eqref{eq:param_dual_problem}, $D^\star_{\bbtheta}$ is the dual optimum, attained by an optimal multiplier $\bblambda^\star$. Both the Lagrangian maximizers $\bbtheta^\dagger(\bblambda)$ and the optimal  multipliers $\bblambda^\star$ are a set in general. In particular, $\bblambda^\star$ will refer to either the set of all optimal dual multipliers or any element of this set, with due clarification when needed.

The dual problem is appealing because finding Lagrangian maximizers is an unconstrained problem. Moreover, the dual function, being the pointwise infimum over affine functions~\cite[p.~81]{boyd2004convex}, is always convex in $\bblambda$. Finally, a subgradient of the dual function at a given $\bblambda$, denoted by $\bbs_{\bbtheta}(\bblambda) \in \partial g_{\bbtheta}(\bblambda)$, is readily available by Danskin's theorem~\cite{danskin2012theory} as
\begin{align} \label{eq:param_dual_subgradient}
\bbs_{\bbtheta}(\bblambda) \coloneqq \bbf \Big( \bbr \big( \bbH, \bbp_{\bbtheta} \big( \bbH; \bbtheta^\dagger(\bblambda) \big)  \Big).
\end{align}
Since the minimization in~\eqref{eq:param_dual_problem} is a convex problem and \eqref{eq:param_dual_subgradient} is a descent direction for the dual function, a solution to the dual problem~\eqref{eq:param_dual_problem} can be obtained by a dual descent method that alternates between unconstrained primal maximization and dual subgradient descent.

\begin{remark}
The duality gap of the parametrized problem $D^\star_{\bbtheta} - P^\star_{\bbtheta}$ is nonzero in general. However, the underlying unparametrized, functional problem, e.g., \eqref{eq:nonparam_problem}, exhibits zero duality gap under mild conditions, and a sufficiently rich parametrization leads to a small duality gap~\cite{elenter2024nearoptimal}. 
\end{remark}

\subsection{Dual Gradient Descent (DGD) Algorithm}

Given an initial dual multiplier $\bblambda_0$, and a dual step size $\eta_{\bblambda}$, a dual (sub)gradient descent algorithm alternates between primal maximization and dual subgradient updates,
\begin{subequations}\label{eq:DGA}
\begin{alignat}{2}
    % \begin{align}
    \bbtheta_k &= \bbtheta^{\dagger}(\bblambda_k) \in \argmax_{\bbtheta \in \bbTheta} \ccalL_{\bbtheta}(\bbtheta, \bblambda_k), \label{eq:DGA:primal} \\
    \bblambda_{k+1} &= \Big[ \bblambda_k - \eta_{\bblambda} \bbs_{\bbtheta}(\bblambda_k) \Big]_{+}, \label{eq:DGA:dual}
    % \end{align}
\end{alignat}
\end{subequations}
\noindent where $[\cdot]_+ \coloneqq \max(\cdot, \bb0)$ is the projection to the nonnegative orthant. Under the assumptions of bounded subgradients and the existence of strictly feasible primal variables, the sequence of iterates $(\bbtheta_k, \bblambda_k)_{k \geq 0}$ generated by the DGD algorithm in \eqref{eq:DGA} is asymptotically feasible, i.e.,
\begin{align} \label{eq:DGD:as-feasibility}
    \lim_{K \to \infty} \frac{1}{K} \sum_{k = 0}^{K-1} \bbf \left( \bbr \Big( \h, \bbp_{\bbtheta} \big(\h; \bbtheta^\dagger(\bblambda_k) \big) \Big) \right) \geq \bb0, \;\, a.s.,
\end{align}
and near-optimal, i.e.,
\begin{align} \label{eq:DGD:near-optimality}
    \lim_{K \to \infty} \! \E \left[ \frac{1}{K} \sum_{k = 0}^{K-1} f_0 \left( \bbr \Big( \h, \bbp_{\bbtheta} \big( \h; \bbtheta^\dagger(\bblambda_k) \big) \Big) \right) \right] \geq D^\star_{\bbtheta} - \ccalO(\eta_{\bblambda}).
\end{align}
This result has been proven several times in~\cite{ribeiro2010ergodic, StateAugmented_RRM_GNN_naderializadeh_TSP2022, calvo2021state} under different setups, hence, we omit the proof. For future reference, we denote by $\Dlambda(\h; \bbtheta^\dagger)$ the dual multiplier trajectories induced by the dual updates in \eqref{eq:DGA:dual}.

We note that the sequence of DGD iterates need not converge to a pair of optimal variables $(\bbtheta^\star, \bblambda^\star)$---optimality and feasibility guarantees hold for the sequential execution of the generated trajectories, not for any individual iterates or their averages (see also Fig.~\ref{fig:policy-switching}). Consequently, the DGD method cannot be stopped at an arbitrary iteration and should be run online to ensure the optimality of trajectories, which is impractical because each dual iteration requires solving a Lagrangian maximization subproblem. The state-augmented learning algorithm introduced in the next section resolves this by learning Lagrangian maximizers for all dual variables during offline training, reducing online computation to lightweight, stochastic dual updates.

\begin{remark}
The optimal solutions to our formulation are in general stochastic policies. The set of achievable ergodic performance vectors under deterministic policies is nonconvex and the Pareto-optimal boundary is attainable only through temporal randomization among deterministic policies. The dual dynamics in~\eqref{eq:DGA} naturally generate this randomization, making ergodic feasibility the appropriate notion for our setting. Per-iteration feasibility is nonetheless achievable when the constraints are convex, as the guarantees in~\eqref{eq:DGD:as-feasibility}--\eqref{eq:DGD:near-optimality} rely on a convex hull relaxation~\cite[Appendix B]{StateAugmented_RRM_GNN_naderializadeh_TSP2022}.
\end{remark}
\section{State-Augmented Resource Allocation with Dual Variable Regression}\label{sec:alg}

\subsection{State-Augmented Parametrization}
A state-augmented parametrization, first proposed in~\cite{calvo2021state}, interprets the dual multipliers as additional inputs to the policy, and is trained to learn Lagrangian maximizing policies for all dual variable values and network configurations jointly. We write a state-augmented policy $\bbp_{\bbphi}: \ccalH \times \reals^N_+ \times \bbPhi \to \reals^N$ as $\bbp_{\bbphi}(\bbH, \bblambda; \bbphi)$ with parameters $\bbphi \in \bbPhi$. The corresponding Lagrangian is given by
\begin{align}
\ccalL_{\bbphi}(\bbphi; \bblambda, \bbH) &= f_0 \Big( \bbr \big(\bbH, \bbp_{\bbphi}(\bbH, \bblambda; \bbphi) \big) \Big) \nonumber \\
  &+ \bblambda \cdot \bbf \Big( \bbr \big(\bbH, \bbp_{\bbphi}(\bbH, \bblambda; \bbphi) \big) \Big).\label{eq:param_lagrangian_state_augmented_over_graph_distribution}
\end{align}%
Considering a joint probability distribution over the dual multipliers and network configurations $\Dlambdah$, we maximize an expectation over the Lagrangian, i.e., 
\begin{align}
\bbphi^{\star} &\in \argmax_{\bbphi \in \bbPhi}\; \E_{(\bblambda, \bbH) \sim \Dlambdah} \big[ \ccalL_{\bbphi}(\bbphi; \bblambda, \bbH) \big]. \label{eq:optimal_state_augmented_over_graph_distribution}
\end{align}
\noindent An optimal $\bbphi^\star$ that solves \eqref{eq:optimal_state_augmented_over_graph_distribution} approximates the resource allocation decisions made by the DGD iterates $\bbtheta^\dagger(\bblambda)$ given in \eqref{eq:DGA:primal}, i.e., we have $\bbp_{\bbphi}(\bblambda_k, \bbH; \bbphi^\star) \approx \bbp_{\bbtheta}\big(\bbH; \bbtheta^\dagger(\bblambda_k) \big)$ and $\ccalL_{\bbphi}(\bbphi^\star; \bblambda_k) \approx \ccalL_{\bbtheta} \big(\bblambda_k; \bbtheta^{\dagger}(\bblambda_k) \big)$, provided that the parameter class $\bbPhi$ is sufficiently expressive and the training distribution $\Dlambdah$ adequately covers the dual multiplier dynamics of~\eqref{eq:DGA:dual}. The former condition is formalized through use of near-universal parametrizations, while for the latter, we propose an effective sampling scheme in Section~\ref{sec:proposed-sampling-and-initialization}.

\begin{definition}[$\nu$-universal parametrizations {\cite[Def.~1]{StateAugmented_RRM_GNN_naderializadeh_TSP2022}}]\label{def:near-universal} 
    A state-augmented parametrization $\bbp_{\bbphi}(\bbH, \bblambda; \bbphi)$ for $\bbphi \in \bbPhi$ is near-universal with degree $\nu > 0$, or equivalently, $\nu$-universal, for functions $\bbp_{\bbtheta}(\bbH; \bbtheta^\dagger(\bblambda))$ if there exists $\bbphi \in \bbPhi$ such that
    \begin{align}
        \hspace{-.5em}\E_{\Dlambda(\h; \bbtheta^\dagger)} \| \bbp_{\bbphi}(\bbH, \bblambda; \bbphi) - \bbp_{\bbtheta}\big(\bbH, \bbtheta^\dagger(\bblambda)\big) \| \leq \nu, \;\, \ccalH\text{-a.e.}%
    \end{align}
    % holds $\ccalH-$.a.e.
\end{definition}
\noindent State-augmented learning \emph{eliminates the need for online computation of optimal primal variables} $\bbtheta^\dagger(\bblambda_k)$ by introducing a single parametrization $\bbphi$. The expected Lagrangian maximization step in~\eqref{eq:optimal_state_augmented_over_graph_distribution} is performed jointly for all possible networks and dual multipliers drawn from $\Dlambdah$ during offline training, while \emph{only the dual dynamics are run during online execution.} We utilize a trained state-augmented policy $\bbphi^\star$ to replace the two-step update sequence in~\eqref{eq:DGA} with (stochastic) updates,
\begin{align}\label{eq:dual_dynamics}
&\hspace{-.5em}\bblambda_{k+1} \!=\! \left[\bblambda_k \!-\! \eta_{\bblambda} \bbf \Bigg( \!\! \frac{1}{T_0} \!\!\! \! \sum_{t=kT_0}^{(k+1)T_0 - 1} \! \! \! \! \! \! \bbr \Big(\bbH_t, \bbp_{\bbphi}(\bbH, \bblambda_k; \bbphi^\star) \Big) \! \Bigg) \! \right]_+ \! \!\!\!.
\end{align}
\noindent In \eqref{eq:dual_dynamics}, trained state-augmented parametrization executes Lagrangian maximizing policies $\bbp_{\bbphi}(\bbH, \bblambda_{\lfloor t / T_0 \rfloor}; \bbphi^\star)$ at each time step $t$ and stochastic dual updates descend the parametrized Lagrangian at every $T_0$ steps. Here, $k\in\{0,1,\dots,K-1\}$ is an iteration index with $K=\lfloor T / T_0 \rfloor$, and $T_0$ denotes the time duration between consecutive dual variable updates. Note that the constraints are evaluated through the $T_0$-averaged instantaneous performance vectors, which serve as stochastic estimates of the ergodic performance [cf.~\eqref{eq:ergodic-performance-function}].

Optimality and feasibility guarantees, analogous to those of the DGD algorithm [cf.~\eqref{eq:DGD:as-feasibility} and \eqref{eq:DGD:near-optimality}], have been proven for the state-augmented trajectories generated by \eqref{eq:dual_dynamics}, as formalized in the following theorem.\footnote{The state-augmented policies in \cite{StateAugmented_RRM_GNN_naderializadeh_TSP2022} operate on the instantaneous network states $\h_t$, whereas our policies operate on the large-scale network configuration $\h$ [cf.~\eqref{eq:ergodic-performance-function}]. However, the theoretical analysis remains the same.}

\begin{theorem}{\cite[Thm.~2]{StateAugmented_RRM_GNN_naderializadeh_TSP2022}} \label{thm:state-augmentation-optimality}
    Assuming the boundedness and strict feasibility of the constraints, expectation-wise Lipschitzness of the network-wide objective utility function $f_0$ and the constraint utility $\bbf$, and the $\nu$-universality of the state-augmented parameterization $\bbp^{\bbphi}(\bbH, \bblambda; \bbphi)$, along with the expected values of the objective and utility constraints within each iteration providing unbiased estimates of the objective and constraints in \eqref{eq:param_problem}, the sequence of decisions made by the state-augmented algorithm in~\eqref{eq:dual_dynamics} is feasible, i.e.,
\begin{align}\label{eq:thm_feasibility_state_augmented}
   \hspace{-0.25cm} \lim_{T\to\infty} \! \bbf \! \left( \! \frac{1}{T}  \sum_{t=0}^{T-1} \bbr \Big[ \bbH_t, \bbp_{\bbphi}\!\left(\bbH, \bblambda_{\lfloor t/T_0 \rfloor}; \bbphi^{\star}\right) \Big] \!\right) \geq \bbzero, \;\; a.s.,
    \end{align}
    and near-optimal, i.e.,
    \begin{align}
    \lim_{T\to\infty} \E & \left[ f_0 \left( \frac{1}{T} \sum_{t=0}^{T-1} \bbr \Big[\bbH_t, \bbp_{\bbphi}\left(\bbH, \bblambda_{\lfloor t/T_0 \rfloor}; \bbphi^{\star}\right) \Big] \right)\right] \nonumber \\
    &\geq D^{\star}_{\bbtheta} - \ccalO(\eta_{\bblambda}) - \ccalO(\nu).
    \label{eq:thm_optimality_state_augmented}
    \end{align}
    % The optimality gap is proportional to the dual update step size and near-universality degree of the parametrization. 
\end{theorem}

Theorem~\ref{thm:state-augmentation-optimality} shows that state-augmented policies are almost surely feasible and near-optimal given sufficient run time, with an additive optimality gap that can be made arbitrarily small through a small dual step size $\eta_{\bblambda}$ and near-universal parametrizations. However, in practical deployments where network dynamics shift quickly, a small $\eta_{\bblambda}$ paired with an improper initialization leads to a prolonged transient of suboptimal decisions. This underscores the need for initializing the dual dynamics close to the optimal multipliers $\bblambda^\star(\bbH)$.

% Theorem~\ref{thm:state-augmentation-optimality} shows that the state-augmented policies are almost surely feasible and near-optimal given long enough run time. Moreover, the objective utility attained is within a constant additive gap of $D^\star_{\bbtheta}$ which can be made arbitrarily small when a sufficiently small dual step size $\eta_{\bblambda}$ and near-universal parametrizations are utilized. Nonetheless, in a practical deployment where the network dynamics $\h_t$ shift quickly over time, using an arbitrarily small dual step size $\eta_{\bblambda}$ with an improper initialization will be ineffective for meeting the optimality and feasibility guarantees in finite time/iterations. This underscores the need for initializing the dual dynamics close to the optimal dual multipliers $\bblambda^\star(\h)$ for a given $\h$. 

\subsection{Proposed Sampling \& Initialization of Dual Multipliers}\label{sec:proposed-sampling-and-initialization}

Departing from prior state-augmented approaches where dual multipliers are drawn from a uniform prior during training and initialized at zero during inference, we propose to (i) train the state-augmented policy $\bbphi$ with dual multipliers sampled from the associated dual descent trajectories for each training network configuration $\bbH$, and (ii) initialize the dual multipliers at their near-optimal levels $\bblambda^\star(\bbH)$ during inference.

Our multiplier sampling scheme periodically rolls out state-augmented dual dynamics at designated training checkpoints and saves the generated trajectories to a buffer, from which dual multipliers are sampled until the next checkpoint. Starting from an initial prior $\Dlambda(\bbH)$ (e.g., uniform), the procedure gradually adjusts the sampling distributions to match the dual descent trajectories $\Dlambda(\bbH; \bbphi)$. This way, we aim to settle on an optimal policy $\bbphi^\star$ with training distributions $\Dlambda(\h) \approx \Dlambda(\h; \bbphi^\star) \approx \Dlambda(\h; \bbtheta^\dagger)$. Implementation and training details are provided in Appendix~\ref{app:sadr_training}.

Given $\bbH$ and a trained state-augmented policy $\bbphi^\star$, we estimate an optimal dual multiplier $\bblambda^\star(\bbH)$ by generating trajectories of dual dynamics and computing the expected time-average of the dual multiplier iterates,
\begin{align}\label{eq:ergodic_average_dual_multiplier}
    \bblambda^\dagger(\bbH; \bbphi^\star) =& \lim_{K \to \infty} \E_{\bblambda_k \sim \Dlambda(\h; \bbphi^\star)} \Bigg[ \frac{1}{K} \sum_{k=0}^{K-1} \bblambda_k(\bbH) \Bigg],
\end{align}
\noindent where we denote by $\bblambda^\dagger(\bbH; \bbphi^\star)$ an optimal dual multiplier estimate, and the expectation is over the randomness of the state-augmented dynamics $\{ \bblambda_k(\h) \}_{k \geq 0}$ in \eqref{eq:dual_dynamics} generated by $\bbphi^\star$ for a given $\h$. In practice, we approximate~\eqref{eq:ergodic_average_dual_multiplier} by rolling out a batch of trajectories for sufficiently large $K$. This allows for a fast initialization of the online dynamics for any given $\h$. In Section~\ref{sec:convergence_results}, we prove that the DGD counterpart of~\eqref{eq:ergodic_average_dual_multiplier}, where the expectation is over $\Dlambda(\h; \bbtheta^\dagger)$, yields an estimate that can be made arbitrarily close to an optimal dual multiplier $\bblambda^\star(\h)$. Assuming $\Dlambda(\h; \bbphi^\star) \approx \Dlambda(\h; \bbtheta^\dagger)$ for an optimal state-augmented policy $\bbphi^\star$, \eqref{eq:ergodic_average_dual_multiplier} serves as a theoretically grounded method for near-optimal initialization.

We next introduce a supervised learning algorithm---which we call the dual variable regression or \emph{dual-regression (DR) algorithm} for short---that succeeds the state-augmented training algorithm and learns $\bblambda^\dagger(\bbH; \bbphi^\star)$ for all $\bbH \sim \Dh$. 
% This way, we learn the near-optimal multipliers for network configurations seen during training and generalize to unseen ones during inference.

\begin{remark}
The computational cost of the proposed dual descent sampling is minimal. We perform dual descent roll-outs at regular checkpoints, not at every training iteration, and each roll-out evaluates constraint slacks with the state-augmented parametrization weights frozen---no backpropagation through the model is required---making the per-roll-out cost negligible relative to a primal training step. This procedure is analogous to standard model validation on held-out networks and effectively repurposes that computation to also generate informative training multipliers.
\end{remark}

\subsection{Dual-Regression Parametrization}

We define a dual-regression function $\bbd_{\bbpsi}: \ccalX \times \ccalH \times \bbPsi \to \reals^c_+$, parametrized by $\bbpsi \in \bbPsi$, where $\ccalX \subseteq \reals^{F_{\bbpsi}}$ denotes the set of regression features. For a given $\bbH$, the regression loss is given by
\begin{align} \label{eq:dual-regression-loss}
    \ccalF \big( \bbX, \bbH, \bbphi^\star; \bbpsi \big) = \big \| \bbd_{\bbpsi}\big( \bbX, \bbH; \bbpsi \big) - \bblambda^\dagger \big( \bbH; \bbphi^\star \big)  \big \|^2_p,
\end{align}
where $\|.\|_p$ denotes the $L_p$ vector norm. Specifically, we optimize a $L_{1}$ loss instead of the $L_{2}$ loss, given the sparsity of the optimal multipliers. 

In \eqref{eq:dual-regression-loss}, near-optimal multiplier estimates $\bblambda^\dagger(\bbH; \bbphi^\star)$ serve as regression targets (labels), and an \emph{optimal dual-regression parametrization $\bbpsi^\star$} minimizes the expected regression loss,
\begin{align} \label{eq:dual_regression_problem:optimal_dual_regressor}
    \bbpsi^\star \in& \argmin_{\bbpsi \in \bbPsi}\; \E_{\bbH \sim \Dh} \left[ \ccalF \big( \bbX(\h), \bbH, \bbphi^\star; \bbpsi \big) \right].
\end{align} 

At the start of the online execution of the state-augmented algorithm, for any given $\bbH$, we utilize the trained dual-regression policy $\bbpsi^\star$ to set the initial dual multipliers close to their optimal levels as
\begin{align} \label{eq:dual-multiplier-regressor-init}
    \bblambda_0(\bbH) := \bbd_{\bbpsi} \big( \bbX(\h), \bbH; \bbpsi^\star \big),
\end{align}
and run the online dual updates in~\eqref{eq:dual_dynamics} as before.

We design the regression features $\bbX(\bbH)$ to be informative of per-user constraint difficulty, and by extension, the magnitude of optimal multipliers $\bblambda^\star(\bbH)$. The specific choice is application-dependent, and we discuss our choice in Section~\ref{sec:experiments} after introducing the power control problem [cf.~\eqref{eq:power_control_problem} and~\eqref{eq:example-dual-regression-feature}].

\begin{remark}
The dual-regression policy is trained after the state-augmented policy as a standard supervised regression problem. This is substantially cheaper than state-augmented training since the dual-regression policy regresses on a single target vector per training network, and the regression targets are already available from the dual multiplier trajectories accumulated during state-augmented training at no additional cost. At inference, the dual-regression policy adds only a single forward pass at the start of execution, amortized over the entire deployment horizon.
\end{remark}

\section{Convergence and Excursion}\label{sec:convergence_results}

We present theoretical insights into the dual dynamics of the DGD algorithm underlying the state-augmented approach. Following the separability argument in Section~\ref{sec:formulation}, we work with a fixed network configuration $\bbH$ throughout this section.

\subsection{Preliminaries}

We first state necessary assumptions and definitions.

\begin{assumption}[Unbiased stochastic subgradients with bounded expected norms]\label{ass:unbiased_and_bounded_subgradients}
    For any $\bblambda$, the stochastic subgradient $\widehat{\bbs}_{\bbtheta}(\bblambda)$ is an unbiased estimate of a true subgradient, i.e., $\E[\widehat{\bbs}_{\bbtheta}(\bblambda)] = \bbs_{\bbtheta}(\bblambda) \in \partial g_{\bbtheta}(\bblambda)$. Furthermore, the second moment of the norm of the subgradients is bounded, i.e., $\E[\norm{\widehat{\bbs}_{\bbtheta}(\bblambda)}_2^2] \leq S^2$ for some $S > 0$.
\end{assumption}

\begin{assumption}[Strictly feasible policy exists.]\label{ass:strict_feasibility}
    For any $\bbH$, there exists a strictly feasible primal variable $\bbtheta^\dagger$ and a constant $\xi > 0$ such that we have $\bbf \left( \bbr \big[\bbH, \bbp_{\bbtheta}(\bbH; \bbtheta^\dagger) \big] \right) \geq \xi \mathbf{1}$.
\end{assumption}

\begin{definition}\label{def:dual-function-near-optimality-ball}
For a given $\epsilon > 0$, we define an $\epsilon$-ball of near-optimality for parametrized dual function iterates as
\begin{align} \label{eq:near-optimal-dual-function-set}
    \ccalG_{\epsilon} \coloneqq \Big\{ \bblambda \in \mathbb{R}^N_{+} \; \cond \; g_{\bbtheta}(\bblambda) - D^\star_{\bbtheta} \leq \epsilon \Big\},
\end{align}
and the $\epsilon$-ball of near-optimal dual multipliers as
\begin{align} \label{eq:near-optimal-dual-multiplier-set}
     \bbLambda_{\epsilon} \coloneqq \Big \{\bblambda \in \mathbb{R}^N_+ \; \cond \; \| \bblambda - \bblambda^\star \| \leq \epsilon \Big \},
\end{align}
where $ \| \bblambda - \bblambda^\star \|$ is understood as the distance of $\bblambda$ to the closest optimal dual multiplier, i.e., to the set $\bbLambda^\star_0$.
\end{definition}

The unbiasedness assumption may not hold exactly in practice since the dual update window $T_0$ is finite. However, for a suitable choice of $T_0$, the bias can be made sufficiently small, and the stochastic subgradients are likely to be descent directions (see~\cite[Appendix II]{calvo2021state}). The bounded second moment assumption is satisfied when the utilities $f_i$ are bounded on their domains.

The following lemma shows that dual multipliers with near-optimal dual function values are bounded.
\begin{lemma}\label{lemma:boundedness_of_D}
Given Assumption~\ref{ass:unbiased_and_bounded_subgradients} and Assumption~\ref{ass:strict_feasibility}, any dual multiplier $\bblambda \in \ccalG_{\epsilon}$ is contained in the ball
\begin{equation}
    \begin{aligned}
        \norm{\bblambda}_1 \leq \frac{D^\star_{\bbtheta} + \epsilon - f_0 \left( \bbr \left( \h, \bbp_{\bbtheta}\big(\h; \bbtheta^\dagger \big) \right) \right)}{\xi},
    \end{aligned}
\end{equation}
where $\xi$ is the feasibility gap and $\bbtheta^\dagger$ are the strictly feasible model parameters as defined in Assumption~\ref{ass:strict_feasibility}.
\end{lemma}
\begin{proof}\let\qed\relax
    See Appendix~\ref{proof:boundedness_of_D}.
\end{proof}

An immediate consequence of Lemma~\ref{lemma:boundedness_of_D} is that the set of optimal dual multipliers $\bbLambda^\star$ is bounded, and by the equivalence of norms, near-optimality in dual function value implies proximity to an optimal multiplier. We formally state these as corollaries.

\begin{corollary}\label{corollary:bounded_optimal_multipliers}
    Under Assumption~\ref{ass:unbiased_and_bounded_subgradients} and Assumption~\ref{ass:strict_feasibility}, the set of optimal dual multipliers $\bbLambda^\star_0$ is bounded.
\end{corollary}

\begin{corollary}\label{corollary:near_optimal_dual_multipliers}
   Under Assumptions~\ref{ass:unbiased_and_bounded_subgradients} and~\ref{ass:strict_feasibility}, for any given $\epsilon > 0$, there exists a $B_{\epsilon}$-ball of near-optimal dual multipliers that contain $\ccalG_{\epsilon}$, i.e., $\ccalG_{\epsilon} \subseteq  \bbLambda_{B_{\epsilon}}$.
\end{corollary} %

The next lemma establishes that the dual dynamics perform stochastic subgradient descent on the dual function.
\begin{lemma}\label{lemma:DGA_expected_lambda_norm_squared}
For any $k \geq 0$, the dual iterates of the stochastic DGD algorithm satisfy
% \begin{subequations}
    \begin{align}
        &\E \left[ \norm{\bblambda_{k+1} - \bblambda^\star}^2 \cond \bblambda_k \right] \nonumber \\
        &\leq \norm{\bblambda_k - \bblambda^\star}^2 +\eta_{\bblambda}^2 S^2 - 2\eta_{\bblambda} \left[ g_{\bbtheta}(\bblambda_k) - D^\star_{\bbtheta} \right].
    \end{align}
% \end{subequations}
\end{lemma}
\begin{proof}\let\qed\relax
    See Appendix~\ref{proof:DGA_expected_lambda_norm_squared}.
\end{proof}

Lemma~\ref{lemma:DGA_expected_lambda_norm_squared} implies that whenever the dual function optimality gaps $g_{\bbtheta}(\bblambda_k) - D^\star_{\bbtheta}$ are larger than $\eta_{\bblambda} \frac{S^2}{2}$, dual dynamics tend to decrease the optimality gap in expectation, and vise versa. We establish in Proposition~\ref{proposition:dual_function_best_convergence} and Theorem~\ref{thm:excursion_bound} that stochastic dual dynamics drive the dual function iterates to regularly visit a near-optimality neighborhood $\displaystyle \ccalG_{\eta_{\bblambda} \frac{S^2}{2}}$, and tend not to drift far between consecutive visits, respectively.

\subsection{Ergodic Complementary Slackness and Near-Optimality}

We characterize the near-optimality of the time-averaged dual iterates that serve as estimates of optimal multipliers.

\begin{proposition}\label{prop:convergence_dual_multiplier_averages}
    Given a network configuration $\bbH$ and initial dual multiplier $\bblambda_0$, let $\{\bblambda_k\}_{k \geq 0}$ be the sequence of dual multipliers generated by the stochastic DGD algorithm in~\eqref{eq:DGA} and define $\bar{\bblambda}_K \coloneqq (1/K)\sum_{k=0}^{K-1} \bblambda_k$ as the time/iteration-average of the dual multipliers up to time/iteration $K$ and define expected dual multiplier averaged over $K$ iterations as
    \begin{align}
        \bblambda^\dagger_{K} \coloneqq \E_{\bblambda_k \sim \Dlambda(\h; \bbtheta^\dagger)} \left[ \bar{\bblambda}_K \cond \bblambda_0 \right],
    \end{align}
    where the expectation is over the trajectories generated by the DGD algorithm [cf.~\eqref{eq:ergodic_average_dual_multiplier}]. Then under Assumption~\ref{ass:unbiased_and_bounded_subgradients} and Assumption~\ref{ass:strict_feasibility}, the dual function evaluated at $\bblambda^\dagger_{K}$ satisfies
    \begin{equation}
        \begin{aligned} \label{eq:ecs_dual_function_averages}
        g_{\bbtheta} \big( \bblambda^\dagger_{K} \big) -  
     D^\star_{\bbtheta}\leq \eta_{\bblambda} \frac{S^2}{2} + \frac{\norm{\bblambda_0 - \bblambda^\star}^2}{2\eta_{\bblambda}K},
        \end{aligned}
    \end{equation}
    and there exists some positive constant $B_{\eta_{\bblambda} \frac{S^2}{2}}$ such that
    \begin{equation}
        \begin{aligned} \label{eq:ecs_dual_multiplier_averages}
            \big \| \bblambda^\dagger_{K} - \bblambda^\star \big \| \leq B_{ \eta_{\bblambda} \frac{S^2}{2}} \left( 1 + \frac{\norm{\bblambda_0 - \bblambda^\star}^2}{K \eta_{\bblambda}^2 S^2} \right).
        \end{aligned}
    \end{equation}
     
\end{proposition}
\begin{proof}\let\qed\relax
    See Appendix~\ref{proof:convergence_dual_multiplier_averages}.
\end{proof}

We gain important insights into the finite-time behavior of the dual dynamics through Proposition~\ref{prop:convergence_dual_multiplier_averages}. Since the optimality gaps are proportional to the dual step size $\eta_{\bblambda}$, we should ideally choose a small step size. However, if the distance of the initial multiplier to the set of optimal multipliers $\norm{\bblambda_0 - \bblambda^\star}$ is large and $\eta_{\bblambda}$ is small, then the second term in the near-optimality bound decays relatively slowly, and the DGD algorithm generates suboptimal decisions for a longer period at the beginning of inference. This presents a trade-off between finite-time optimality and feasibility while also offering an opportunity for significant improvement by initializing the dual multipliers close to their optimal levels. 

By taking the liminf of both sides in \eqref{eq:ecs_dual_multiplier_averages}, it follows that
\begin{equation}
    \begin{aligned}
        \liminf_{K \to \infty}\; \big \|  \bblambda^\dagger_K - \bblambda^\star \big \| \leq B_{\eta_{\bblambda} \frac{S^2}{2}}.
    \end{aligned}
\end{equation}
Assuming the limit $\bblambda^\dagger = \lim_{K \to \infty} \bblambda^\dagger_K$ exists, the expected time-averaged dual multiplier $\bblambda^\dagger$ satisfies $|| \bblambda^\dagger - \bblambda^\star || \leq B_{\eta_{\bblambda} \frac{S^2}{2}}$, or equivalently, $\bblambda^\dagger \in \bbLambda_{B_{\eta_{\bblambda} \frac{S^2}{2}}}$.
Similarly, the expected time average of the dual function increments converges to $\ccalG_{\eta_{\bblambda} \frac{S^2}{2}}$ in the limit, i.e.,
\begin{equation}
    \begin{aligned}
        \liminf_{K \to \infty}\;  g_{\bbtheta} \big( \bblambda^\dagger_{K} \big) - D^\star_{\bbtheta} \leq \eta_{\bblambda} \frac{S^2}{2}.
    \end{aligned}
\end{equation}
To restate, both the expected time-averaged multiplier and the dual function converge to near-optimal neighborhoods in the limit. However, these convergence claims hold in expectation and do not rule out arbitrarily bad realizations of dual descent trajectories. The convergence and excursion results stated next show that this is fortunately not a concern. For almost every realization, the dual function iterates (resp. dual multiplier iterates) enter a small neighborhood of $D^\star_{\bbtheta}$ (resp. $\bbLambda^\star_0$) infinitely often and deviations therefrom are highly unlikely.

\subsection{Almost-sure Convergence of Best Iterates}

Before stating the stochastic convergence result, it is instructive to analyze the deterministic case. Assume access to an oracle that replaces the stochastic subgradients with their unbiased expectations. By an appeal to recursion, we have
 \begin{align} \label{eq:dual_function_best_convergence_with_oracle}
        \min_{k=0, \ldots, K} \big[ g_{\bbtheta}(\bblambda_k) - D^\star_{\bbtheta} \big]
        \leq \frac{\norm{\bblambda^\star - \bblambda_0}^2 + K\eta_{\bblambda}^2 S^2}{2\eta_{\bblambda} K}.
\end{align}
The convergence of the best dual function optimality gap follows since the right hand side of~\eqref{eq:dual_function_best_convergence_with_oracle} approaches $\eta_{\bblambda} \frac{S^2}{2}$ as $K \to \infty$. Since the stochastic subgradients $\widehat{\bbs}(\bblambda)$ oscillate around a true subgradient, it is reasonable to expect a similar convergence behavior in the stochastic case. The following proposition from~\cite[Theorem~1]{ribeiro2010ergodic} confirms this intuition by leveraging ideas from supermartingale theory.

\begin{proposition}\label{proposition:dual_function_best_convergence}
    Consider the DGD algorithm given in~\eqref{eq:DGA}. Given any $\bblambda_0$ and that Assumption~\ref{ass:unbiased_and_bounded_subgradients} holds, the best dual function iterate $\dbest{K} \coloneqq \min_{k = 0, \ldots, K} \big[ g_{\bbtheta}(\bblambda_k) \big]$ satisfies
    \begin{align}
        \liminf_{K \to \infty}\; \dbest{K} - D^\star_{\bbtheta} \leq \eta_{\bblambda} \frac{S^2}{2} \quad \text{a.s.} 
    \end{align}
\end{proposition}
\begin{proof}
    See Appendix~\ref{proof:dual_function_best_convergence}.
\end{proof}

Proposition~\ref{proposition:dual_function_best_convergence} asserts that for almost all realizations and arbitrary $\delta > 0$, $g_{\bbtheta}(\bblambda_k) \leq D^\star_{\bbtheta} + \eta_{\bblambda} \frac{S^2}{2} - \delta$ occurs eventually for some $k$. Moreover, by a time-shift argument, $g_{\bbtheta}(\bblambda_k) \leq D^\star_{\bbtheta} + \eta_{\bblambda} \frac{S^2}{2} - \delta$ happens infinitely many times. By Corollary~\ref{corollary:near_optimal_dual_multipliers}, an analogous statement holds for the dual iterates, i.e., 
\begin{align}\label{corollary:dual_iterate_best_convergence}
    \liminf_{K \to \infty}\; \lambdabest{K} \leq B_{\eta_{\bblambda} \frac{S^2}{2}} \quad \text{a.s.},
\end{align}
where $\lambdabest{K} \coloneqq \min_{k = 0, \ldots, K}  \big \| \bblambda_k - \bblambda^\star \big \| $.

\subsection{Excursions of Dual Function Iterates}
Proposition~\ref{proposition:dual_function_best_convergence} guarantees that the dual function iterates visit a near-optimality neighborhood, but does not characterize the behavior between two consecutive visits. We refer to deviations from this neighborhood as excursions and derive probability bounds on their magnitude. 

Before stating the main theorem, we formally define an excursion and state additional assumptions.

\begin{definition}[Excursion of optimality gaps]
    \label{definition:excursion}
    For $k \geq 0$, let $G_k \coloneqq g_{\bbtheta}(\bblambda_k) - D^\star_{\bbtheta}$ be the dual function optimality gaps as before and define a stopping time $L \coloneqq \min \{l>0 \, \vert \, G_{k_0 + l} < G_{k_0} \}$. We define the sequence $\{G_{k_0}, \ldots, G_{k_0 + L} \}$ as an excursion of the dual function optimality gaps and $G^{\dagger}_{k_0} = \max_{0 \leq l \leq L} G_{k_0 + l}$ as the worst excursion gap.
\end{definition}

\begin{assumption}[Local Lipschitz-smoothness and strong-convexity] \label{ass:lipschitz_smooth_strongly_convex_subregion}
    There exist constants $\lambda_{\max}, G_{\max}, m, L > 0$ and an excursion subregion $\ccalE$ such that $\| \bblambda^\star \| \leq \lambda_{\max}$ for any $\bblambda^\star$ and the dual function is $L$-Lipschitz smooth on $\ccalE$, i.e., 
    \begin{align} \label{eq:lipschitz-smooth}
        \hspace{-.3em}g_{\bbtheta}(\bblambda_2) \!-\! g_{\bbtheta}(\bblambda_1) &\leq  
        \bbs_{\bbtheta} (\bblambda_1)^\top (\bblambda_2 - \bblambda_1) \!+\! \frac{L}{2} \| \bblambda_2 - \bblambda_1 \|^2,
    \end{align}
    $\forall \bblambda_1, \bblambda_2 \in \ccalE$ and $m$-strongly convex on $\ccalE$, i.e.,
    \begin{align} \label{eq:strong-convexity}
        \hspace{-.5em}g_{\bbtheta}(\bblambda_2) \!-\! g_{\bbtheta}(\bblambda_1) &\geq  \bbs_{\bbtheta} (\bblambda_1)^\top (\bblambda_2 - \bblambda_1) \!+\! \frac{m}{2} \| \bblambda_2 - \bblambda_1 \|^2,
    \end{align}
    $\forall \bblambda_1, \bblambda_2 \in \ccalE$, where the excursion subregion $\ccalE$ is defined as
    \begin{align} \label{eq:nice-excursion-subregion}
        \ccalE \coloneqq \big \{  \bblambda \in \reals^c_+ \cond G_{\min} \leq g_{\bbtheta}(\bblambda) - D^\star_{\bbtheta} \leq G_{\max} \big \},
    \end{align}
    where $G_{\min} \coloneqq \eta_{\bblambda} \frac{S^2}{2} \left( \frac{L}{2m} \right)$ denotes the inner excursion boundary and we note that it is a constant multiple of $\eta_{\bblambda} \frac{S^2}{2}$. 
\end{assumption}

Assumption~\ref{ass:lipschitz_smooth_strongly_convex_subregion} requires local regularity of the dual function within the excursion subregion $\ccalE$ so that subgradients grow sufficiently fast with the optimality gap. This is considerably milder than assuming global strong convexity and Lipschitz smoothness. Next, we state the main excursion theorem.

\begin{theorem}[Excursion Bound for Dual Function Optimality Gaps]\label{thm:excursion_bound}
Consider a sequence of dual function optimality gaps $\{ G_k \}_{k \geq 0}$ with $G_0 = (1 + \rho) G_{\min}$ for any $\rho > 0$, and suppose Assumptions~\ref{ass:unbiased_and_bounded_subgradients}-\ref{ass:lipschitz_smooth_strongly_convex_subregion} hold. Then the probability of $\gamma$-excursions of the dual function optimality gaps is upper bounded by
\begin{align}
     \mathbb{P} \big[  G^{\dagger}_{0} \geq \gamma \; \vert \; G_0 \big] &\leq \frac{ e^{-\beta \big( \min \big \{  \gamma, G_{\max}  \big \} - G_0 \big) } }{ 1 + \min \big \{  0,  \gamma - G_{\max} - 1 \big \} },
 \end{align}  
 with $\beta \coloneqq \frac{\rho L}{S^2  \left( 1 + \rho^2 \eta_{\bblambda}^2 \frac{L^2}{4} \right) } > 0$.
\end{theorem}

\begin{proof}
See Appendix~\ref{proof:excursion_bound}.
\end{proof}
The proof exploits the supermartingale property of the optimality gaps $G_k$. A linear bound is due to Markov's inequality, while the exponential bound follows from an exponential transform that preserves the supermartingale property within $\ccalE$. Together with Proposition~\ref{proposition:dual_function_best_convergence}, the excursion result establishes that the dual dynamics regularly visit a near-optimality neighborhood and are unlikely to drift far between visits.
\section{Numerical Results: Optimal Power Control}
\label{sec:experiments}

We test the proposed state-augmented learning algorithm in a power control problem in a setting similar to~\cite{StateAugmented_RRM_GNN_naderializadeh_TSP2022}. To briefly summarize the setup, all wireless network realizations are sampled from a family of network configurations comprised of $N$ users, i.e., transmitter-receiver pairs $\{(\mathsf{Tx}_i, \mathsf{Rx}_i)\}_{i=1}^N$, dropped uniformly within an $R \times R$ square area with an average user density of $\rho = N / R^2$ users/\SI{}{\kilo\meter\squared}. Each transmitter communicates with a single designated receiver, and each receiver treats the incoming interference from all but its associated transmitter as noise. The (instantaneous) network state at time step $t$,  $\bbH_t\in\mathbb{R}^{N \times N}_+$, and the large-scale (slow-varying) network state, i.e., network configuration, $\bbH \in \mathbb{R}^{N \times N}_+$, contain the magnitude of all complex channel gains, where the instantaneous and long-term channel gains between transmitter $\mathsf{Tx}_i$ and receiver $\mathsf{Rx}_j$ are denoted by $h_{ij, t}, h_{ij} \in \mathbb{C}$.

Denoting the maximum transmit power by $P_{\max}$, resource allocation vectors $\bbp \in [0, P_{\max}]^N$ represent the transmit power levels of the transmitters. Given network configuration $\h$ and resource allocation vector $\bbp = \bbp(\h)$, the instantaneous performance vector $\bbr \big(\h_t, \bbp \big)$ represents the receiver rates, and the rate for each receiver $\mathsf{Rx}_i$ at timestep $t$ is given by
\begin{align} \label{eq:rate}
    \!\!r_i(\bbH_t, \bbp) = \log_2\left(1+\frac{p_i \left|h_{ii,t}\right|^2}{W N_0 + \sum_{j=1, j\neq i}^N p_j \left|h_{ji,t}\right|^2}\right),
\end{align}
where $W$ and $N_0$ denote the channel bandwidth and noise power spectral density (PSD), respectively, and it is assumed that capacity-achieving codes are used. The ergodic performance vector $\bbr \big( \h, \bbp(\h) \big) = (1/T) \sum_{t = 0}^{T-1} \bbr \big( \h_t, \bbp(\h) \big)$ accordingly represents the ergodic receiver rates and is evaluated over a horizon of $T$ time steps [cf.~\eqref{eq:ergodic-performance-function}]. We can write the power control problem for any given $\h$ as
\begin{subequations}\label{eq:power_control_problem}
\begin{alignat}{2}
    P^\star(\h) = &\maximum_{\bbp(\bbH)} &~& \frac{1}{T} \sum_{t=0}^{T-1} \mathbf{1}_N^\top \bbr \big( \bbH_t, \bbp(\bbH) \big),     \\
    &\st && \frac{1}{T} \sum_{t=0}^{T-1}  \bbr \big( \bbH_t, \bbp(\bbH) \big) \geq \mathbf{1}_N f_{\min},%
\end{alignat}
\end{subequations}
\noindent where we imposed a minimum ergodic rate requirement of $f_{\min}$ bps/Hz for all receivers by considering a minimum-rate constraint utility $\bbf(\bbx) = \bbx - \mathbf{1}_N f_{\min}$ and a sum-rate objective utility $f_0(\bbx) = \mathbf{1}^\top_N \bbx$.
% \footnote{In all of our plots, we rescale the Lagrangian and objective values by the number of users $N$ and therefore report the mean-rate utility for the latter.} 
We note that both the power control and dual regression tasks involve node-level decisions; therefore we seek permutation-equivariant policies.

\subsection{Graph Neural Network (GNN) Parametrizations}
\label{sec:gnn-parametrization}

Graph neural networks (GNNs) have been widely used in learning resource allocation policies for wireless networks owing to their permutation equivariance, transferability, and stability properties~\cite{eisen2020optimal, naderializadeh2022learning, lee2020graph, shen2019graph}. As such, we parametrize both state-augmented (SA) policies and dual-regression (DR) policies by GNN architectures. For brevity, we refer to these two models as primal-GNN and dual-GNN, respectively.

We represent the network configuration $\bbH$ as a graph $\ccalG_\bbH = (\ccalV, \ccalE, \bbW, \ccalY)$. The nodes $\ccalV = \{1, \ldots, N\}$ correspond to users (transmitter-receiver pairs). We denote by $\ccalE = \ccalV \times \ccalV$ the set of directed edges, and $\bbW: \ccalE \mapsto \reals$ weighs each edge from node $i$ to node $j$ by log-normalized channel gain between $i$th transmitter and $j$th receiver, i.e., $e_{ij} = \frac{1}{Z}\log_2\big(1 + \frac{P_{\max}|h_{ij}|^2}{W N_0}\big)$ where $Z > 0$ is a constant. The graphs are directed since $h_{ij} \neq h_{ji}$ in general. We sparsify by removing edges with $e_{ij} < 0.01$, yielding non-uniformly connected graphs whose node degrees vary with the local geometry and fading realization. Note that this sparsification affects only the GNN aggregation and the SINR in~\eqref{eq:rate} still sums over all interferers. 

We denote by $\ccalY$ the node features supported on these graphs. For the primal-GNN, the node features vary over time and at each time step $t$, we set them to the dual multipliers $\ccalY_t = \bblambda_{\lfloor t/T_0 \rfloor}(\bbH)$. For the dual-GNN, we always set the input node features to the DR features, i.e., $\ccalY = \bbX(\bbH)$.

We work with the weighted adjacency matrix $\bbW \in \reals^{|\ccalV| \times |\ccalV|}$ with $W_{ij} = e_{ij}$ for all $i, j \in \ccalV$. Both GNNs share the same backbone and process graph data through a cascade of $L$ graph convolutional network (GCN) layers, each mapping input node embeddings $\bbZ^{(\ell-1)}$ to output node embeddings by
\begin{equation}\label{eq:gcn_layer}
\hspace{-0.1cm}\bbZ^{(\ell)} \!=\! \bbPsi^{(\ell)}\big(\bbZ^{(\ell-1)}; \bbS, \bbTheta^{(\ell)}\big) = \varphi\! \left[\sum_{k=0}^{K} \bbS^k \bbZ^{(\ell-1)} \bbTheta^{(\ell)}_k\right]\!.
\end{equation}
In~\eqref{eq:gcn_layer}, $\Theta^{(\ell)} = \{\bbTheta^{(\ell)}_k \in \reals^{F_{\ell-1} \times F_\ell}\}_{k=0}^K$ is a set of learnable graph filter coefficients, and $\varphi$ is a pointwise nonlinearity, e.g., ReLU. The policy parametrizations $\bbphi$ and $\bbpsi$ include all the respective weights $\{\Theta^{(\ell)}\}_{1 \leq \ell \leq L}$, plus any other learnable weights such as those of normalization layers.

We use $L = 3$ layers, $F_\ell = 64$ hidden channels, and $K=2$-hop aggregations for both models. The dual-GNN feeds regression features directly as input embeddings, while the primal-GNN first passes the dual multipliers through a sinusoidal embedding composed with an MLP. The output dimension is $F_L = 1$ for both models, with a sigmoid activation for the primal-GNN to ensure $\bbp \in [0, P_{\max}]^N$, and a negated log-sigmoid scaled by $\lambda_{\max} = 50$, set conservatively, for the dual-GNN to ensure nonnegative and bounded predictions. 
The dual-GNN is invoked once at the start of inference, while the primal-GNN is called at every time step. A single forward pass through either architecture has polynomial cost in the network size, making them suitable for large-scale deployments.

\begin{remark}
Our SA+DR framework with GNN-parametrized policies is applicable to other resource allocation setups. In applications such as link scheduling, it is primarily the action space and the functional form of the utility and constraints that changes, while the graph structure, the ergodic-average constraint formulation, and the dual dynamics with their policy-switching behavior remain operationally the same. The ablated state-augmented algorithm (without the improvements proposed here) has already been successfully applied to link scheduling~\cite{garcia2025linkscheduling} and network slicing~\cite{uslu2024learning}. We leave the extension of the SA+DR pipeline to these settings for future work.
\end{remark}

% We represent the wireless network configuration $\h$ (resp. instantaneous network state $\h_t$) as a fully connected graph,
% where nodes correspond to users, i.e., transmitter-receiver pairs, and the edge weights between nodes $i$ and $j$ are log-normalized long-term channel gains (resp. instantaneous channel gains), i.e., $e_{ij} \propto \log_2 \left( 1 + \frac{P_{\max} |h_{ij}|^2}{WN_0} \right)$. Primal- and dual-GNNs take the dual multiplier vector $\bblambda_k(\h)$ at a given dual iteration $k$ and regression feature matrix $\bbX(\h)$ as input node signals (features) and output the vector of power allocations $\bbp(\h)$ and near-optimal initial dual multipliers $\bblambda_0(\h) \approx \bblambda^\star(\h)$, respectively. We discuss the details of the GNN architectures in the Supplementary Materials.

\begin{figure*}[t!]
    \centering
    \includegraphics[width = \linewidth]{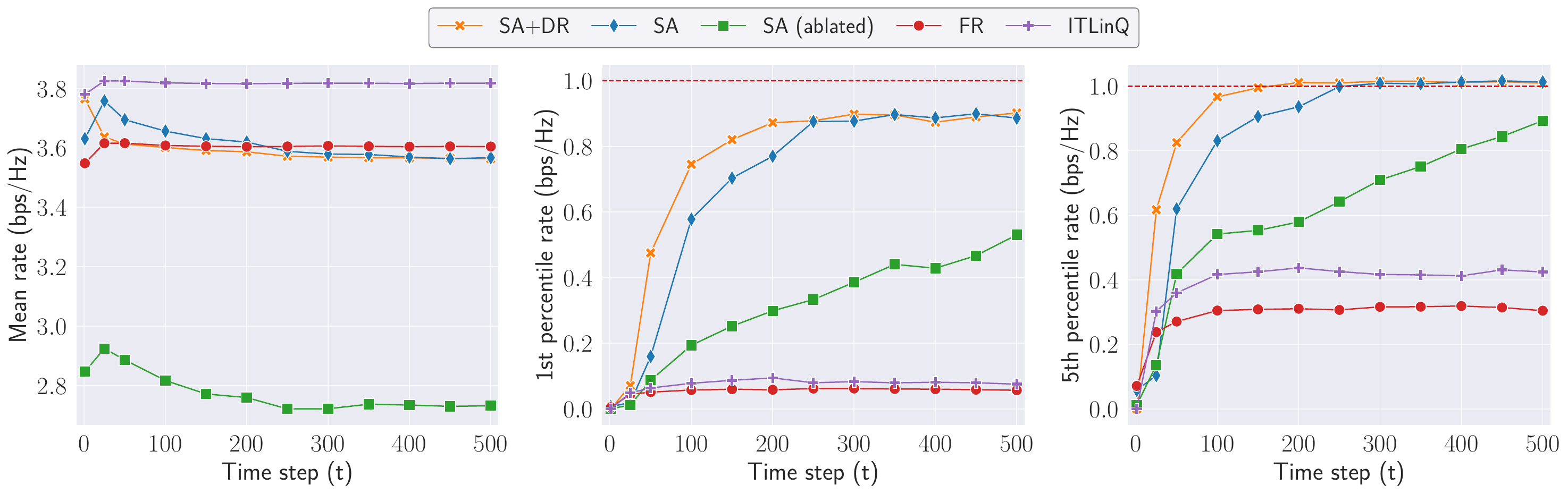}

    \caption{Comparison of time evolution of mean, 1st, and 5th percentiles of all receiver ergodic rates for all the algorithms.
    }
    \label{fig:test_evolution_over_time}
\end{figure*}

\begin{figure*}[t!]
    \centering
    \includegraphics[width=\linewidth]{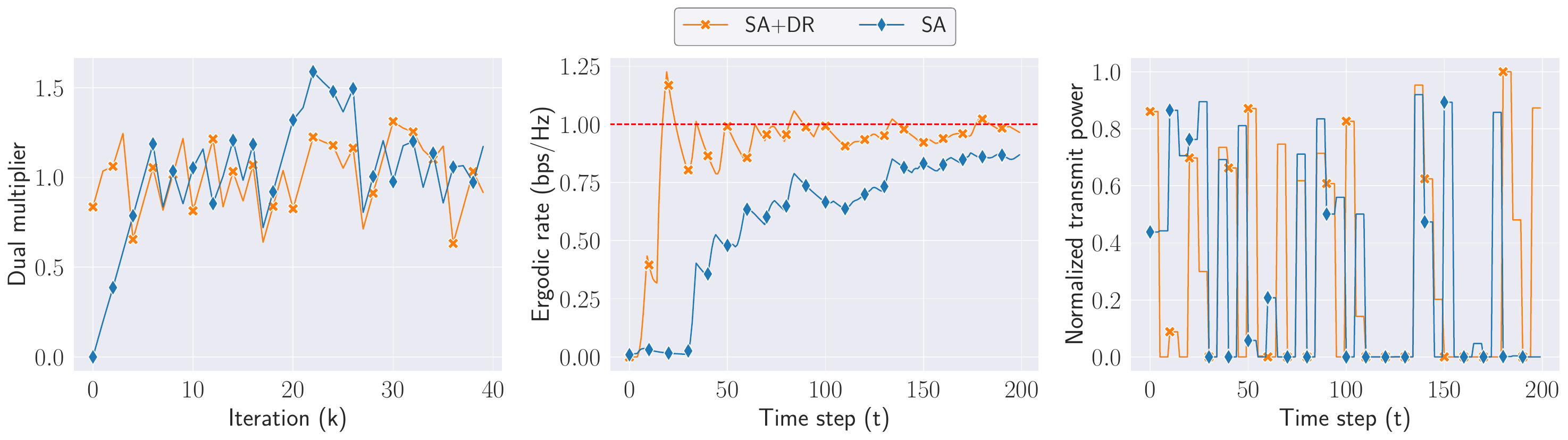}
    \caption{Comparison of dual multipliers, constraints and power allocations over time steps with (SA+DR) and without (SA) near-optimal initialization for an example test network user.}
    \label{fig:example-test-network}
\end{figure*}

\subsection{Ablations \& Baseline Methods}
We compare our proposed state-augmented algorithm with dual variable regression (SA+DR) with (i) the proposed state-augmented algorithm with training dual multipliers sampled from the dual dynamics and zero initialization (SA), (ii) ablated state-augmented algorithm of~\cite{StateAugmented_RRM_GNN_naderializadeh_TSP2022} (SA-ablated), i.e., the state-augmented algorithm with dual multipliers sampled from a uniform prior and zero initialization. For fair comparison, proposed and ablated state-augmented algorithms all share the same GNN backbones. We also compare against two non-learning-based baselines, namely FR and ITLinQ~\cite{naderializadeh2014itlinq, naderializadeh2017ultra}. The FR method employs a fixed full power transmission policy for all transmitters at each time step, whereas in ITLinQ, a subset of transmitters is scheduled at each time to transmit with full power based on an information-theoretic optimality condition for treating interference as noise. 

For the power control problem, we also use the full-reuse (FR) baseline rates as regression features,
\begin{align}\label{eq:example-dual-regression-feature}
    \bbX(\bbH) = \bbr(\bbH, \bbp_{\mathrm{FR}}) = \bbr \big( \bbH, \mathbf{1}_N P_{\max} \big),
\end{align}  
where a single-shot estimation of the ergodic performance function is performed by allocating a flat $P_{\max}$ amount of resources uniformly across the network. The FR rates require only a single evaluation of the rate function under a trivial power allocation, and as shown in Fig.~\ref{fig:dr-scatter}, they correlate inversely with optimal dual multiplier magnitudes.

\subsection{Experiment Configuration}

We briefly mention the experiment parameters, the majority of which carry over from~\cite{StateAugmented_RRM_GNN_naderializadeh_TSP2022} unchanged. The network operates over $T = 200$ time steps and each time step lasts $10$ milliseconds. We consider a dual-slope path-loss model with log-normal shadowing with standard deviation $\sigma = 7.0$ for the large-scale fading and Rayleigh distribution with pedestrian velocity of $1$ m/s for the small-scale fading. Unless otherwise noted, we set $P_{\max} = 10$ dBm, $N_0=-174$ dBm, $W = 20$ MHz, $f_{\min} = 1.0$ bps/Hz. We draw network realizations from a family of configurations with $N = 100$ users and network area side length $R$ chosen suitably to obtain an average user density of $\rho = 8$ users/\SI{}{\kilo\meter\squared}. We update the multipliers every $T_0 = 5$ time steps with a step size of $\eta_{\bblambda} = 0.2$. 

We generate $128$ network realizations in total and split them into training, validation, and test datasets of size $|\ccalD_{\h}| = 128$, $|\ccalV_{\h}| = 16$, $|\ccalT_{\h}| = 64$, respectively. We use 3-layer GNN architectures with hidden layer features of $F_\ell = 64, \ell = 1, 2, 3$ for both primal and dual models. We set the learning rate for both models at $\eta_{\bbphi} = \eta_{\bbpsi} = 10^{-3}$. We refer to Appendix~\ref{app:sadr_training} for additional training details, including pseudocode, training plots, and guidelines on important design parameters.

\subsection{Test Performance of the Proposed SA+DR Algorithm}

Fig.~\ref{fig:test_evolution_over_time} showcases the performance of our proposed method (SA+DR) against the ablated methods and the baselines over the test dataset $\ccalT_{\h}$. For the sake of this figure, the operation window is extended to $T = 500$ time steps, and ergodic rates are computed over a moving window of $200$ time steps. Note that all state-augmented algorithms significantly outperform the baselines in terms of $1{\text{st}}$ and $5{\text{th}}$ percentile rate metrics by eventually (almost) satisfying all the minimum rate constraints at the expense of a small degradation in sum-rate objective values. In particular, the proposed SA+DR algorithm provides near-feasible rates to all users within fewer time steps than the SA algorithm; the time it takes for the 1st and 5th percentile rates to come close to $f_{\min}$, e.g., within 5\%-10\% of $f_{\min}$, is roughly halved. Although the ablated state-augmented algorithm gradually improves the percentile rates---an effect that can be improved for all state-augmented algorithms by increasing the dual step size---it exhibits a noticeable degradation in the objective utility.  

\begin{figure}[ht!]
    \centering
    \includegraphics[width=\linewidth]{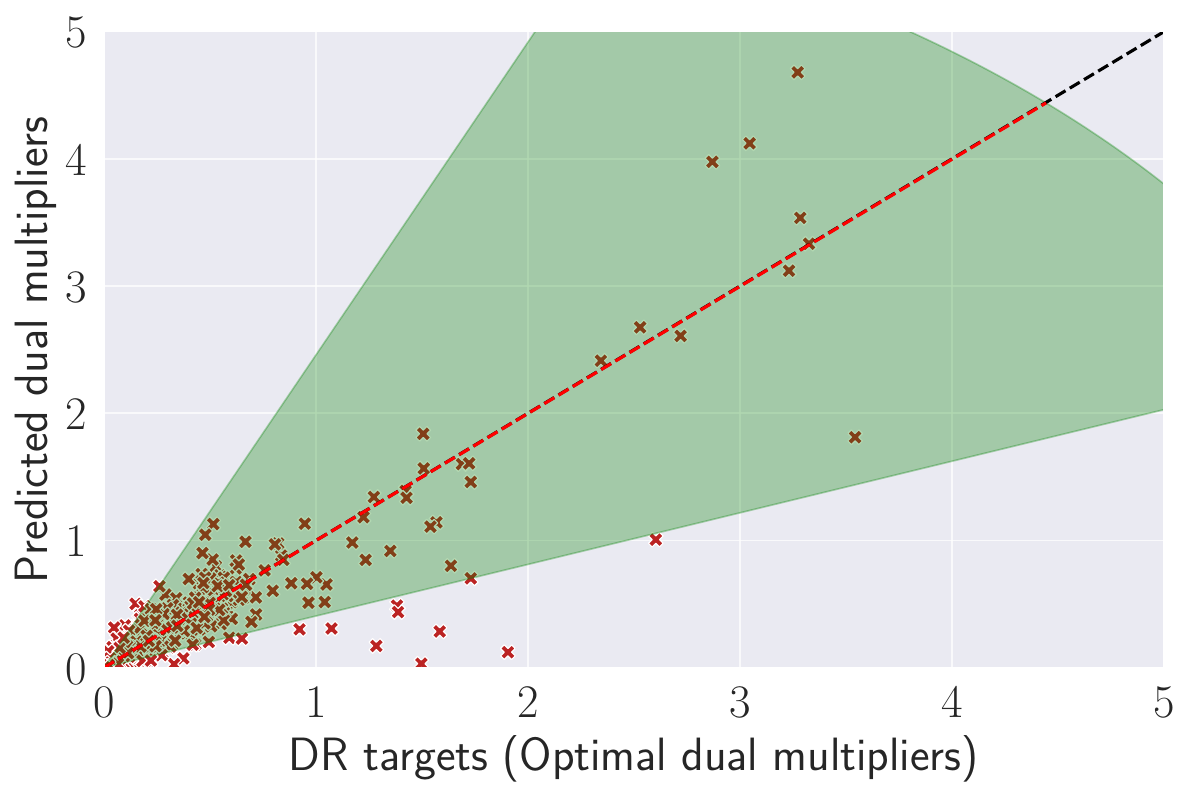}
    
    \caption{Scatter plot of target (optimal) dual multipliers against predictions of the trained dual-GNN. Green cone covers a 95\% confidence interval and red dashed line is the center line of the cone. Black dashed line corresponds to the $y=x$ line.}
    \label{fig:dr-scatter}
\end{figure}

\begin{figure*}[ht!]
    \centering
    \begin{minipage}{.64\linewidth}
    \includegraphics[width = \linewidth]{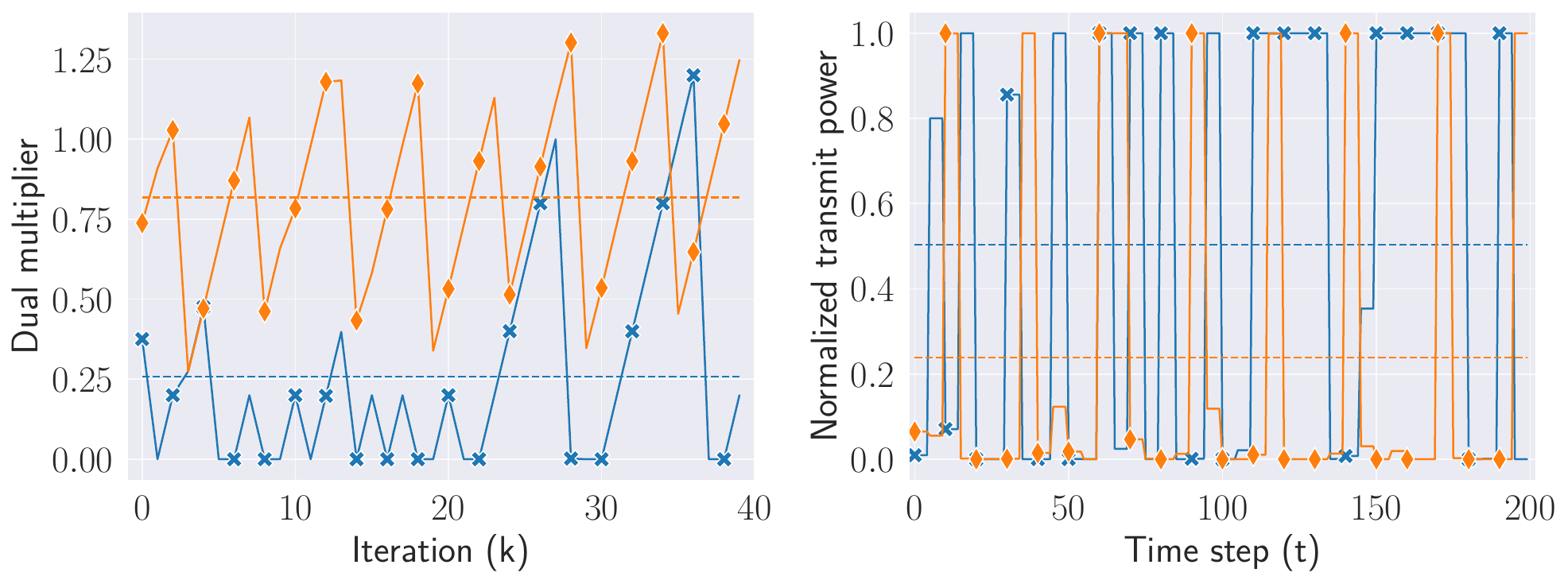}
    \vspace{-1.5em}
        \subcaption[]{}
        % \vspace{-.5em}
    \end{minipage}
    \hfill
    \begin{minipage}{.35\linewidth}
    \includegraphics[width = \linewidth]{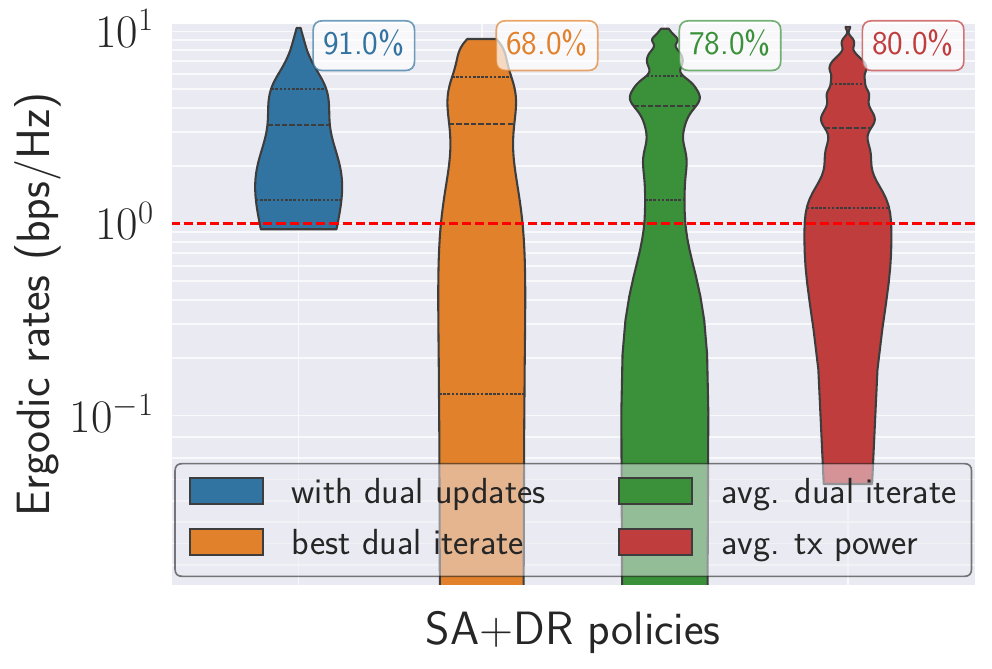}
    \vspace{-1.5em}
        \subcaption[]{}
    \end{minipage}
    \hfill
    \vfill 
    \caption{(a) Example policy switching behavior of two neighboring (mutually-interfering) users. Dashed lines indicate the time / iteration averages. (b) Distribution of the ergodic rates obtained under SA+DR dynamics (shown in blue) with those obtained under three different kinds of fixed policies: (i) SA+DR policy executed for the best dual iterate without dual updates (shown in orange), (ii) SA+DR policy executed for the average dual multiplier without any dual updates (shown in green), and (iii) fixed average transmission power of the SA+DR policy (shown in red). Best dual iterate is the one for which maximum constraint violation is lowest. 
    Computation of the best and average iterates in (i)--(iii) are done for the trajectory given by SA+DR algorithm. We also report the percentage of users that meet $f_{\min}$ under each policy.
    }
    \label{fig:policy-switching}
\end{figure*}

\begin{figure*}[ht!]
    \centering
    \includegraphics[width = \linewidth, height = .28\linewidth]{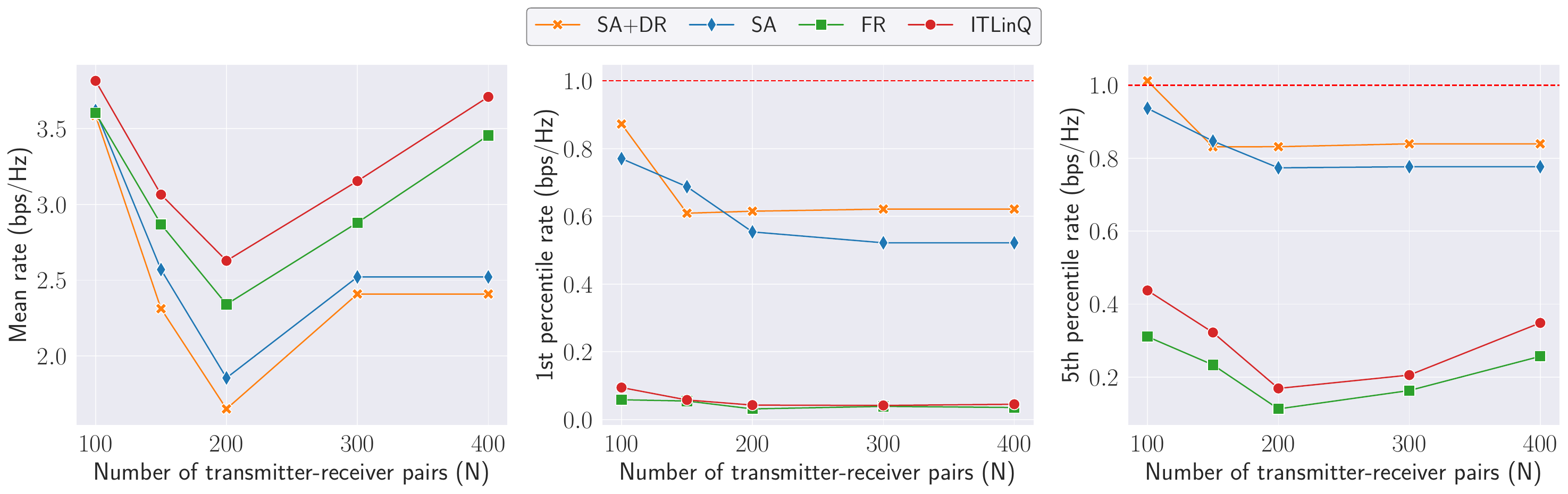}
    \caption{Transferability of the  policies learned on networks with $N = 100$ transmitter-receiver pairs to  larger networks, for various values of $N$ up to $400$. The average user density $\rho$ is the same (fixed) across all $N$.}
    \label{fig:transferability}
\end{figure*}

In Fig.~\ref{fig:example-test-network}, we plot the time evolution of the $n$th component of the dual multiplier vector $\bblambda_{\lfloor t/ T_0\rfloor}$, ergodic rate vector $(1/t) \sum_{\tau = 0}^{t-1} \bbr\left (\h_{\tau}, \bbp_{\bbphi}\big( \h, \bblambda_{\lfloor \tau / T_0 \rfloor}; \bbphi^\star \big) \right)$, and normalized transmit power vector $\bbp_{\bbphi}\big( \h, \bblambda_{\lfloor t / T_0 \rfloor}; \bbphi^\star \big) / P_{\max}$ for an example test user/node $n \in [N]$. Whenever the rate constraints are violated (resp. satisfied) for a given receiver, the corresponding dual multipliers increase (resp. decrease) and the dual dynamics drive both state-augmented policies to gradually improve the ergodic rates of all receivers by adapting the transmission powers based on the inputted dual multipliers. For the user plotted in particular, the dual dynamics do not converge to a single (optimal) dual multiplier and the resulting transmission policies meet the minimum rate requirement by alternating between high and low transmission modes. 

The policy switching behavior depicted in Fig.~\ref{fig:example-test-network} is crucial for guaranteeing the ergodic feasibility of users with the most difficult channel conditions and is further illustrated in Fig.~\ref{fig:policy-switching}. We underscore the close overlap between the SA+DR and SA transmission policies in the later phase of the operation window (after $t = 30$) and how near-optimal dual multiplier initialization of the SA+DR policy converges to a stochastic on/off policy much earlier than without the initialization. Finally, Fig.~\ref{fig:dr-scatter} showcases the close overlap between the DR targets and the optimal dual multipliers predicted by the trained dual-GNN for a subset of validation networks. 

% Next, we take a more in-depth look at the policy switching behavior of the state-augmented algorithms. 

\subsection{Policy Switching \& Infeasibility of Primal-Dual Iterates}
Fig.~\ref{fig:policy-switching}(a) showcases the dual multiplier and transmit power evolutions over time for the SA+DR algorithm, but this time for two different users from an example test network. These two nodes/users correspond to two neighboring transmitter-receiver pairs that cause significant interference to one another. During any given iteration window, only one of the two transmitters is active, i.e., transmits with high power. This way, associated receivers of active transmitters enjoy high instantaneous rates and since the transmitters take turns to transmit, both receivers maintain ergodic rates greater than $f_{\min}$. Although we plotted only the SA+DR algorithm, the same conclusions would hold also for the SA algorithm without dual multiplier initialization. Fig.~\ref{fig:policy-switching}(b) compares the SA+DR policy with several (deterministic) policies derived from state-augmented dynamics. The SA+DR policy significantly benefits from the continuous execution of dual updates---even if near-optimal initialization is achieved.
% ---and policy randomization in terms of near-feasibility of all ergodic rates.  

\subsection{Transferability of GNN-Parametrized SA+DR Policies}
One of the main benefits of GNN parametrizations is that they scale well to graphs of different sizes. In particular, we can train GNN-parametrized policies over small- to medium-sized graphs and evaluate them on larger graphs. In Fig.~\ref{fig:transferability}, our proposed GNN-parametrized SA+DR policy trained on networks of size $N = 100$ and average user density of $\rho = 8$ users/\SI{}{\kilo\meter\squared}, is evaluated on network configurations of varying sizes with the same average user density. Trained primal- and dual-GNN models both exhibit notable transferability to unseen network configurations in this ``fixed-density'' scenario where channel statistics mostly remain similar across scales. 
\section{Conclusion}\label{sec:conclusion}
We tackled a stochastic resource optimization problem in wireless networks. We proposed a state-augmented policy with dual variable regression (SA+DR) for fast sampling of near-optimal and feasible resource allocations. Leveraging GNN parametrizations in a power control application, we demonstrated the scalability and transferability of the learned policies. We supported extensive experimental validation with convergence and excursion theorems for the dual descent dynamics. 
The extension of the SA+DR framework to other classes of resource allocation problems is a natural next step. 

% \appendices
% \appendix

\begin{appendices}

\renewcommand{\thesubsection}{\Alph{section}.\arabic{subsection}}

\setcounter{subsection}{0}

% Supplementary Materials moved to the Appendices in the final manuscript

\section{Implementation \& Training of SA+DR Policies}
\label{app:sadr_training}

\subsection{Training State-Augmented (SA) Policies}
 
In practice, given access to a dataset of i.i.d. network configuration samples $ \bbH_{b}$ for $b = 0, \ldots, B-1$ drawn from $\Dh$ and a corresponding sequence of instantaneous network states $\bbH_{b, t}$ at time steps $t \in \{0, \ldots, T-1 \}$, we replace the expectations over $\Dh$ with empirical (batch) averages. Similarly, given a (conditional) distribution of training dual multipliers $\Dlambda(\h)$, we draw batches of training dual multipliers and network configurations jointly from $\Dlambdah = \Dlambda(\h) \Dh$ and write the empirical Lagrangian corresponding to \eqref{eq:param_lagrangian_state_augmented_over_graph_distribution} as
\begin{align}\label{eq:empirical_lagrangian}
    \hspace{-.25cm}\widehat{\E}_{\Dlambdah} \big[ \Lphi (\bbphi; \bblambda, \bbH) \big] \approx \frac{1}{BM}
    % \sum_{b=0}^{B-1} \sum_{ m=0}^{M-1}
    \sum_{b,m} \ccalL_{\bbphi}\Big( \bbphi; \bblambda_m(\h_b), \bbH_b \Big).
\end{align}
Given initial model parameters $\bbphi_0$ and a step size $\eta_{\bbphi}$, we evaluate \eqref{eq:empirical_lagrangian} over sampled batches of dual multipliers and training network configurations at each training epoch $n=0, 1, \ldots, N_{\mathrm{SA}}-1$, 
and replace the maximization step in \eqref{eq:optimal_state_augmented_over_graph_distribution} with gradient-ascent updates on the model parameters, i.e.,
\begin{align} \label{eq:state-augmented-gradient-descent}
    \bbphi_{n+1} &= \bbphi_{n} + \eta_{\bbphi} \nabla_{\bbphi} \widehat{\E}_{\Dlambdah} \big[ \ccalL_{\bbphi}(\bbphi_n; \bblambda, \bbH) \big].
\end{align}

We note that the performance of optimal model parameters $\bbphi^\star$ can be drastically affected by the distribution of dual multipliers considered. Prior state-augmented learning works \cite{calvo2021state, StateAugmented_RRM_GNN_naderializadeh_TSP2022} contended with a uniform dual multiplier distribution $\Dlambda(\h) = \mathrm{Uniform}(\bb0, \mathbf{1})$ for all $\h$ and observed desirable performance in their experiments. More generally, we can choose a uniform prior distribution $\mathrm{Uniform}(\bb0, \bblambda_{\max})$ and the hyperparameter $\bblambda_{\max}$ can be tuned and/or adjusted dynamically (possibly independently for each training network configuration $\h_b$) to match the dual dynamics trajectories as the training progresses. To do so, we train the state-augmented model parameters with a given initial prior distribution $\Dlambda(\h) = \mathrm{Uniform}(\bb0, \bblambda_{\max})$ until the first checkpoint epoch $N^{0}_{\mathrm{SA}}$. We choose $\bblambda_{\max}$ as a fixed hyperparameter. Given $\bbphi_{N^{0}_{\mathrm{SA}}}$, we roll out the state-augmented dual dynamics, save the trajectories to a buffer, and sample dual multipliers from the buffer $\bblambda(\h_b) \sim \Dlambda(\h_b; \bbphi_{N^{0}_{\mathrm{SA}}})$ for each $\h_b$ until the next checkpoint. In general, at each training epoch $n$, we jointly draw dual multipliers and network configurations as
\begin{align} \label{eq:state-augmentation-training-sampler}
    (\bblambda_m(\h_b), \h_b) \sim \Dlambdah \coloneqq \Dlambda \left( \h; \bbphi_{ \lfloor n / N^{0}_{\mathrm{SA}} \rfloor . n} \right)\Dh,
\end{align} 
and update $\bbphi_{n+1}$ according to~\eqref{eq:state-augmented-gradient-descent} with $\Dlambdah$ given in \eqref{eq:state-augmentation-training-sampler}.

\begin{figure}
\begin{minipage}[t!]{.98\linewidth}
\begin{algorithm}[H]
\caption{Offline Training of the SA+DR Algorithm}
\label{alg:SA-train}
\begin{algorithmic}[1]
\STATE {\bfseries Input:} Training networks $\{\bbH_b\}_{b=0}^{B-1}$; initial prior $\Dlambda(\bbH) = \mathrm{Uniform}(\bb0, \bblambda_{\max})$; maximum and checkpoint epochs $N_{\mathrm{SA}}, N_{\mathrm{DR}}$, and $N^0_{\mathrm{SA}}$; step sizes $\eta_{\bbphi}, \eta_{\bbpsi}, \eta_{\bblambda}$; update window $T_0$.

\STATE Initialize primal-GNN parameters $\bbphi_0$ and sampling distributions (buffers) $\Dlambda(\bbH_b) \gets \mathrm{Uniform}(\bb0, \bblambda_{\max})$ $\forall b$.

\medskip
\STATE \textit{// Phase 1: State-augmented (SA) training}
\FOR{$n = 0, \ldots, N_{\mathrm{SA}} - 1$}
    \STATE Run stochastic gradient ascent on primal-GNN parameters with mini-batch of networks $\{\bbH_b\}$ and, $M$ dual multipliers $\{\bblambda_m(\bbH_b)\}_{m=0}^{M-1} \sim \Dlambda(\bbH_b)$ sampled for each $\bbH_b$.
    \STATE $\bbphi_{n+1} = \bbphi_n + \eta_{\bbphi} \nabla_{\bbphi} \widehat{\E}_{\Dlambdah} \big[ \ccalL_{\bbphi}(\bbphi_n; \bblambda, \bbH) \big]$ \hfill $\triangleright$~\eqref{eq:state-augmented-gradient-descent}
    \IF{$(n+1) \mod N^0_{\mathrm{SA}} = 0$}
        \FOR{each training network $\bbH_b$}
            \STATE Roll out dual dynamics~\eqref{eq:dual_dynamics} with frozen $\bbphi_{n+1}$ for $K = \lfloor T / T_0 \rfloor$ iterations.
            \STATE Update $\Dlambda(\bbH_b) \gets \Dlambda(\bbH_b; \bbphi_{n+1})$. \hfill $\triangleright$~\eqref{eq:state-augmentation-training-sampler}
        \ENDFOR
    \ENDIF
\ENDFOR
\STATE Set $\bbphi^\star \gets \bbphi_{N_{\mathrm{SA}}}$.

\medskip
\STATE \textit{// Phase 2: Regression target estimation}
\FOR{each $\bbH_b$}
    \STATE $\displaystyle \widehat{\bblambda}^\dagger(\bbH_b; \bbphi^\star) \gets \E_{\bblambda \sim \Dlambda(\bbH_b; \bbphi^\star)}[\bblambda]$ \hfill $\triangleright$~\eqref{eq:empirical-regression-target}
\ENDFOR
\STATE Compute DR features $\bbX_b = \bbX(\bbH_b)$ $\forall b$, e.g., via~\eqref{eq:example-dual-regression-feature}.

\medskip
\STATE \textit{// Phase 3: Dual-regression (DR) training}
\STATE Initialize dual-GNN parameters $\bbpsi_0$.
\FOR{$n = 0, \ldots, N_{\mathrm{DR}} - 1$}
    \STATE Run mini-batch gradient descent on dual-GNN parameters.
    \STATE $\bbpsi_{n+1} \!=\! \bbpsi_n \!-\! \eta_{\bbpsi} \nabla_{\bbpsi} \widehat{\E}_{\bbH} \Big[ \ccalF \big( \bbX(\bbH), \bbH, \bbphi^\star; \bbpsi_n \big) \Big]$ \hfill $\triangleright$~\eqref{eq:dual-regression-gradient-descent}
\ENDFOR
\STATE Set $\bbpsi^\star \gets \bbpsi_{N_{\mathrm{DR}}}$.

\medskip
\RETURN{Trained parameters $\bbphi^\star$ and $\bbpsi^\star$.}
\end{algorithmic}
\end{algorithm}
\end{minipage}
\end{figure}

\begin{figure}
\begin{minipage}[t!]{.98\linewidth}
\centering
\begin{algorithm}[H]
\caption{Online Execution of the SA+DR Algorithm}
    \label{alg:SA-test}
    \begin{algorithmic}[1]
    \STATE {\bfseries Input:} Optimal SA and DR parametrizations $\bbphi^\star$, $\bbpsi^\star$, a network configuration $\h$ and corresponding sequence of network states $\{ \bbH_{t} \}_{t=0}^{T-1}$, dual-regression features $\bbX$, update window $T_0$, step size $\eta_{\bblambda}$.
    \STATE Initialize $k \gets 0$ 
, $\bblambda_0 \gets \bbd_{\bbpsi} \big(\bbX, \bbH; \bbpsi^\star\big)$.
    
    \STATE Run the dual dynamics of \eqref{eq:dual_dynamics}, i.e.,
    { \FOR {$t = 0, \ldots, T-1$}
        \STATE Generate state-augmented decisions $\bbp_{\bbphi}(\bbH, \bblambda_k; \bbphi^\star)$.
        \IF {$t+1 \mod T_0 = 0$} 
        \STATE $\bblambda_{k+1} = \Big[ \bblambda_k - $\\
        $\eta_{\bblambda} \bbf \left( \frac{1}{T_0} \sum_{t=kT_0}^{(k+1)T_0 - 1} \bbr \big( \bbH_t, \bbp_{\bbphi}\big(\bbH, \bblambda_k; \bbphi^\star\big) \big) \right) \Big]_+$.
            % \STATE $\bblambda_{k+1} = \Big[ \bblambda_k - \eta_{\bblambda} \bbf \big( \bbH_t, \bbp_{\bbphi}\big(\bbH, \bblambda_k; \bbphi^\star) \big) \Big]_+$
            \STATE $k \gets k+1$.
        \ENDIF
   \ENDFOR }
   
    \RETURN{Allocation decisions $\{ \bbp_{\bbphi}(\bbH, \bblambda_{\lfloor t/T_0 \rfloor}; \bbphi^\star) \}_{t=0}^{T-1}$}
\end{algorithmic}
\end{algorithm}

\end{minipage}
\end{figure}

\subsection{Training Dual-Regression (DR) Policies}

To estimate the dual multiplier regression targets by \eqref{eq:ergodic_average_dual_multiplier} in practice, we roll out $M$ trajectories of the state-augmented dual dynamics over $K$ iterations and set
\begin{align} \label{eq:empirical-regression-target}
    \widehat{\bblambda}^\dagger(\bbH_b; \bbphi_{N_{\mathrm{SA}}}) = \frac{1}{M K} \sum_{m = 0}^{M - 1} \sum_{k = 0}^{K-1} \bblambda_{k, m}(\bbH_b), 
\end{align}
for each $\bbH_b$, $b = 0, \ldots, B-1$. Similar to \eqref{eq:empirical_lagrangian}, we compute the regression loss function \eqref{eq:dual-regression-loss} with regression targets of \eqref{eq:dual_regression_problem:optimal_dual_regressor}, averaged over a batch of training network configurations as
\begin{align} \label{eq:empirical-dual-regression-loss}
    &\widehat{\E}_{\bbH} \Big[ \ccalF \big( \bbX(\h), \bbH, \bbphi_{N_{\mathrm{SA}}}; \bbpsi) \big) \Big] \nonumber \\
    &= \frac{1}{B} \sum_{b=0}^{B-1} \; \big \| \bbd_{\bbpsi}\big( \bbX_b(\h_b), \bbH_b; \bbpsi \big) - \widehat{\bblambda}^\dagger \big( \bbH_b, \bbphi_{N_{\mathrm{SA}}} \big) \big \|^2_1.
\end{align}
\noindent Given initial dual-regression model parameters $\bbpsi_0$, step size $\eta_{\bbpsi}$, and training epochs $n = 0, 1, \ldots, N_{\mathrm{DR}} - 1$, we update the model parameters as
\begin{align} \label{eq:dual-regression-gradient-descent}
    \bbpsi_{n+1} = \bbpsi_n - \eta_{\bbpsi} \nabla_{\bbpsi} \widehat{\E}_{\Dh} \Big[ \ccalF \big( \bbX(\h), \bbH, \bbphi_{N_{\mathrm{SA}}}; \bbpsi) \big) \Big].
\end{align}

% \subsection{Pseudoalgorithms}
We summarize the offline training and online execution procedures in Algorithm~\ref{alg:SA-train} and Algorithm~\ref{alg:SA-test}, respectively. Algorithmically, dynamic adjustment of the dual multiplier training distribution, solving a dual regression problem~\eqref{eq:dual_regression_problem:optimal_dual_regressor} after state-augmented training, and the near-optimal initialization of the dual multipliers [cf.~\eqref{eq:dual-multiplier-regressor-init}] during the execution phase, are the primary enhancements we introduce to the state-augmented algorithm proposed in prior work~\cite{StateAugmented_RRM_GNN_naderializadeh_TSP2022}. Figures~\ref{fig:train_dual}--\ref{fig:train_primal} provide insights into the proposed SA+DR training algorithm.

\begin{figure}[ht!]
        \centering
        \includegraphics[width = .9\linewidth]{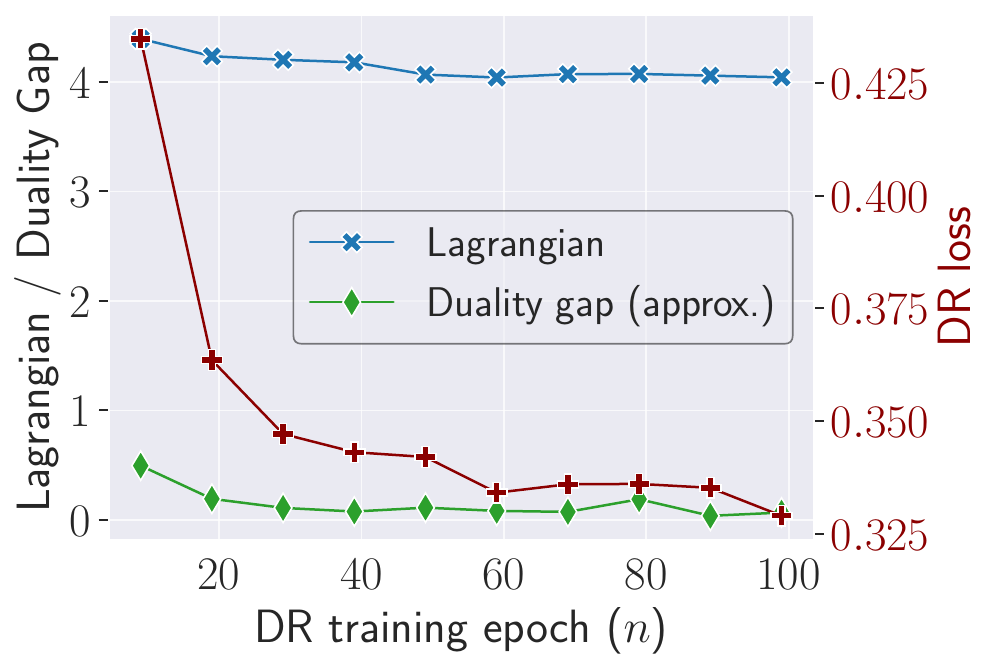}
    \caption{Training progress of the DR model. For each training epoch $n$, we plot the DR loss (shown in dark red), the approximate duality gap and the state-augmented Lagrangian (averaged across all validation networks and corresponding dual multipliers predicted by the DR model) evaluated for the fully-trained SA model $\bbphi_{N_{\mathrm{SA}}}$ and DR model $\bbpsi_{n}$. For visualization, we omitted the early epochs of the DR training. 
    % where DR loss was significantly larger.
    }
    \label{fig:train_dual}
% \end{figure*}
\end{figure}

\begin{figure}[ht!]
    \centering
    \begin{subfigure}[t]{\linewidth}
        \includegraphics[width = \linewidth]{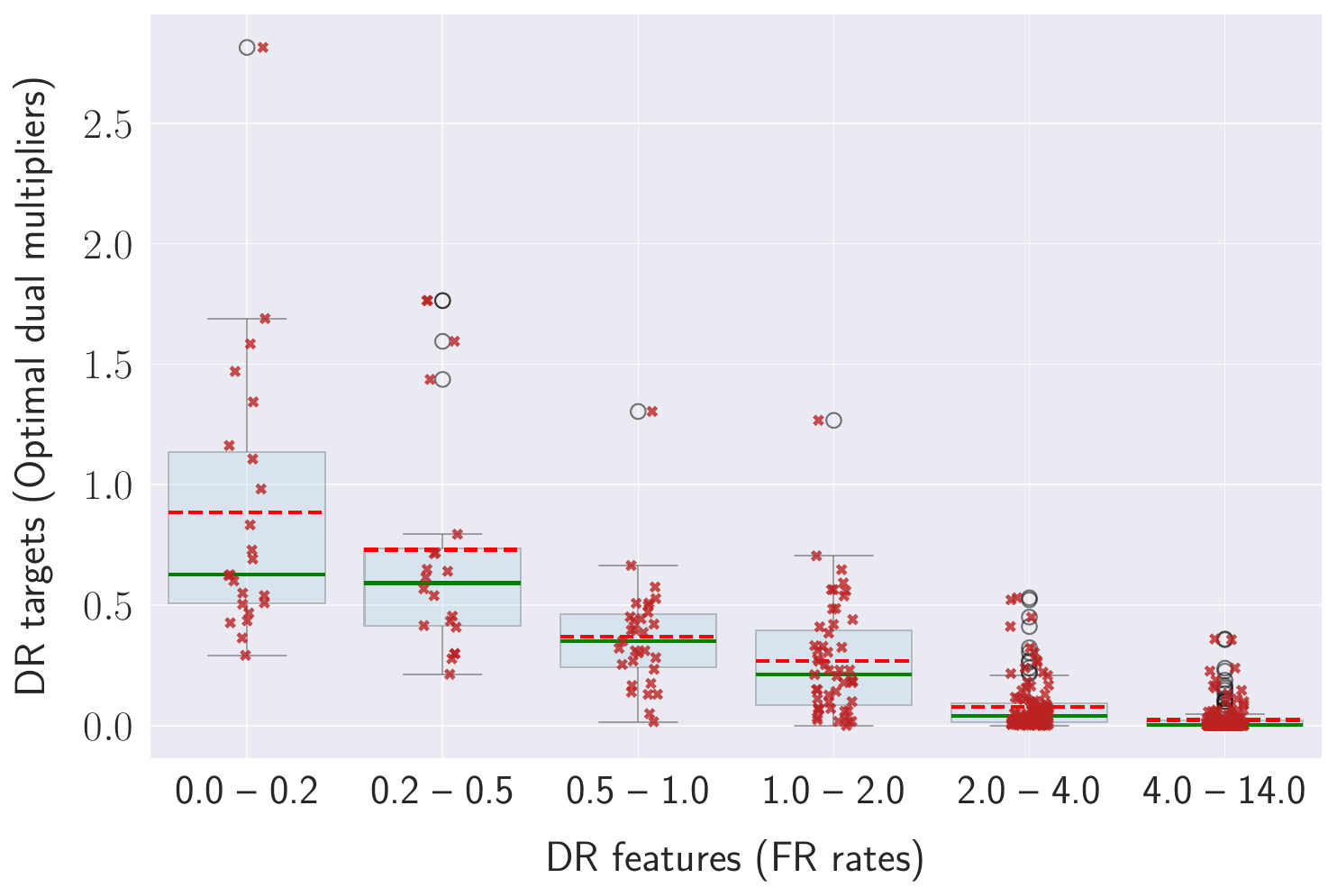}
    \end{subfigure}%
    \caption{Scatter-plot of DR features against the DR targets. The optimal dual multipliers, and their estimates, tend to be larger for interference-limited users with weaker channel gains, which benefit the most from DR initialization.
    }
    \label{fig:train_dual_features_scatter}
\end{figure}

\begin{figure*}[ht!]
    \hfill
    \centering
    \begin{subfigure}[t]{.33\linewidth}
        \centering
        \includegraphics[width = \linewidth]{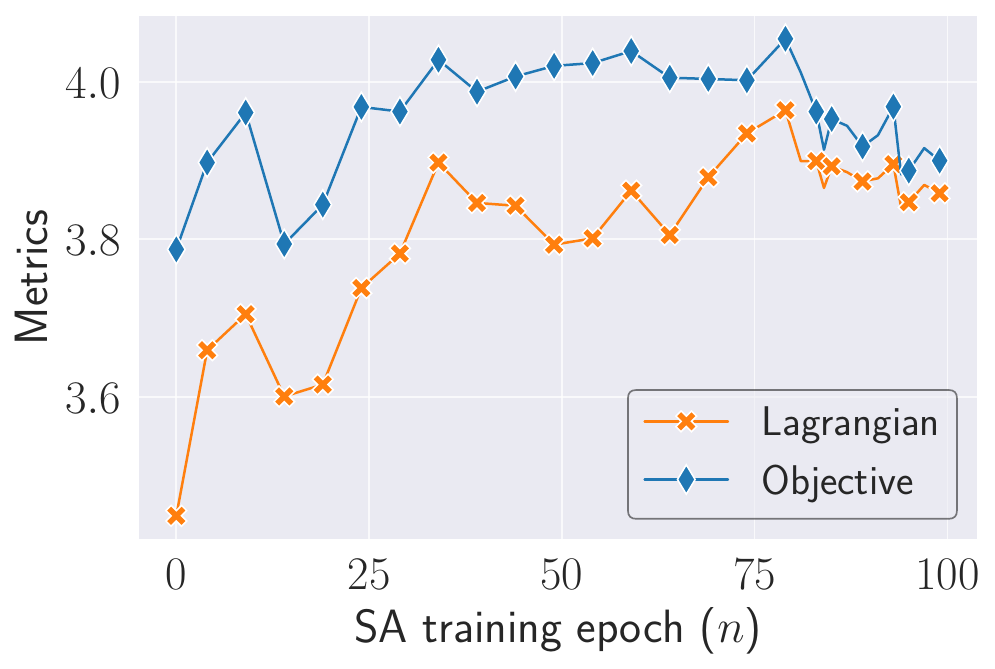}
        % \subcaption{}
    \end{subfigure}%
    \hfill
    \centering
    \begin{subfigure}[t]{.33\linewidth}
        \centering
        \includegraphics[width = \linewidth]{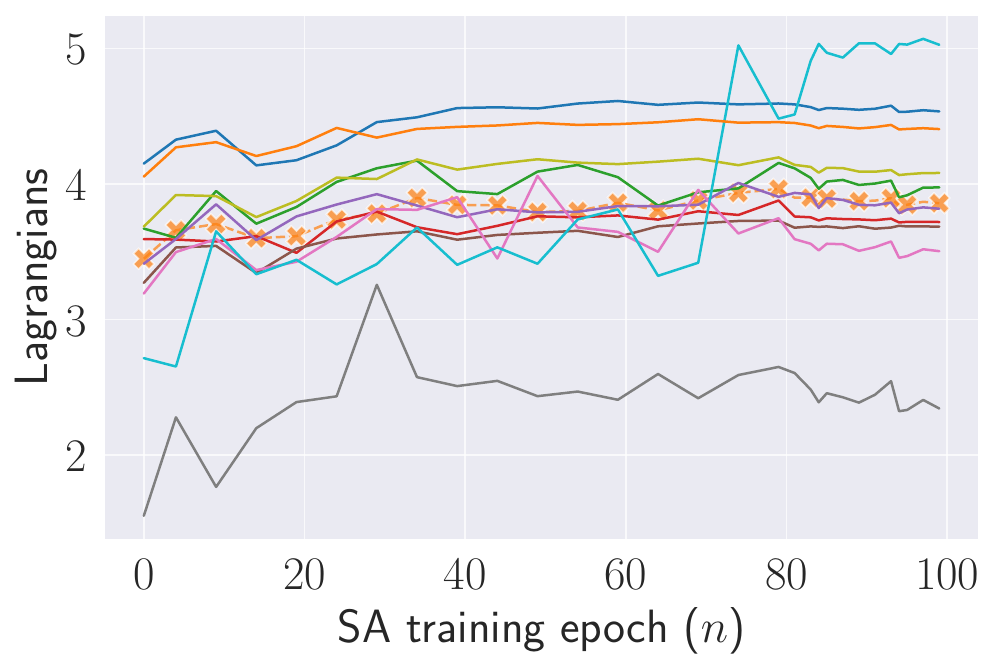}
        % \subcaption{}
    \end{subfigure}%
    \hfill
    \centering
    \begin{subfigure}[t]{.33\linewidth}
        \centering
        \includegraphics[width = \linewidth]{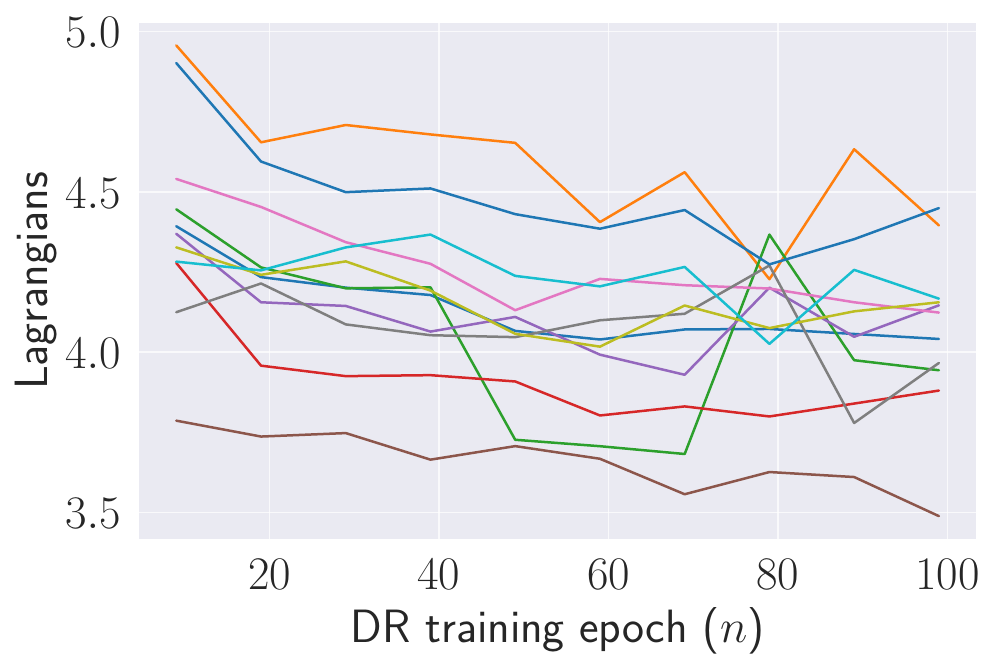}
        % \subcaption{}
    \end{subfigure}%
    \caption{Progress of SA+DR training, evaluated on the validation set. Left: the state-augmented Lagrangian and the objective at the estimated optimal dual multipliers; their difference, i.e., the approximate duality gap, shrinks as training proceeds. Middle: per-network state-augmented Lagrangians (averaged over dual multipliers) being maximized as the SA model trains. Right: the same Lagrangians, now evaluated at the DR-predicted duals, are minimized as DR training progresses.
    % Progress of the SA+DR training algorithm. Leftmost plot shows the state-augmented Lagrangian and the objective evaluated for the estimated optimal dual multipliers in the validation dataset. We note that the approximate duality gap, given by the difference of the two curves, shrinks as the training proceeds. Middle plot shows the maximization of the state-augmented Lagrangians for individual networks (averaged over dual multipliers) in the validation dataset as the SA model is trained. Finally, rightmost plot shows the minimization of the state-augmented Lagrangians for individual networks evaluated for the dual multipliers predicted by the DR model as the DR training progresses. 
    }
    \label{fig:train_primal}
\end{figure*}

\subsection{Other Practical Considerations \& Training Details}
We first train the primal-GNN for $N_{\mathrm{SA}} = 100$ epochs on the training dataset $\ccalD_{\h}$. Each training epoch iterates over $\ccalD_{\h}$ with a mini-batch size of $B_{\h} = 8$ networks, i.e., a total of $100 \times 128/8 = 1600$ iterations, and we sample $B_{\bblambda}$ = 4 dual multipliers for each network. The dual multiplier sampling distribution is initially set to a standard uniform prior, and we initialize an empty dual multiplier buffer for each training network configuration. The buffer capacity is set to 100 dual multiplier vectors per network configuration. We test the primal-GNN on both the training and validation datasets every $N^{0}_{\mathrm{SA}} = 2$ epochs. We roll out the dual dynamics with the same choices of $T, T_0$ and $\eta_{\bblambda}$ as the ones we used in the actual test phase. That is, we roll out the dynamics for $T = 200$ time steps, where the dual update (iteration) window is set to $T_0 = 5$, and the dual step size is set at $\eta_{\bblambda} = 0.2$. The initial multipliers are set to an average multiplier we compute across the buffer of each network configuration, which we use as a proxy for an optimal dual multiplier. Then, the dual multiplier trajectories are added to the buffers corresponding to each network configuration, and we draw dual multipliers from this dynamically evolving buffer as the training progresses.

Training of the primal-GNN is succeeded by the optimal dual multiplier estimation phase and the training of the dual-GNN. We split the training dataset further into training and validation datasets with a $7:1$ ratio. For each given $\h$, the regression targets, $\bblambda^\dagger(\h; \bbphi^\star)$, are estimated by averaging dual multipliers in the corresponding dual multiplier sampling buffer. Given the regression features in \eqref{eq:example-dual-regression-feature} and targets, we also train the dual model for $N_{\mathrm{DR}} = 100$ epochs with a batch size of $B_{\h} = 32$. Regarding the dual-regression training parameters, we observed that, compared to an $L_2$ (MSE) loss, the $L_1$ loss places greater emphasis on the accurate prediction of the large (impactful) multipliers while respecting the sparsity. We also note that our experiments are not sensitive to the value of $\lambda_{\max}$, which was set conservatively at $\lambda_{\max} = 50.0$.

For state-augmented inference, the dual step size $\eta_{\bblambda}$ is the most critical hyperparameter as it governs the trade-off between asymptotic objective optimality, and finite-time feasibility of the constraints. In Fig.~\ref{fig:dual-step-size-ablation}, we show the impact of various dual multiplier step sizes. 

\begin{figure*}[ht!]
    \centering
    
    \begin{minipage}{.64\linewidth}
    \includegraphics[width = \linewidth]{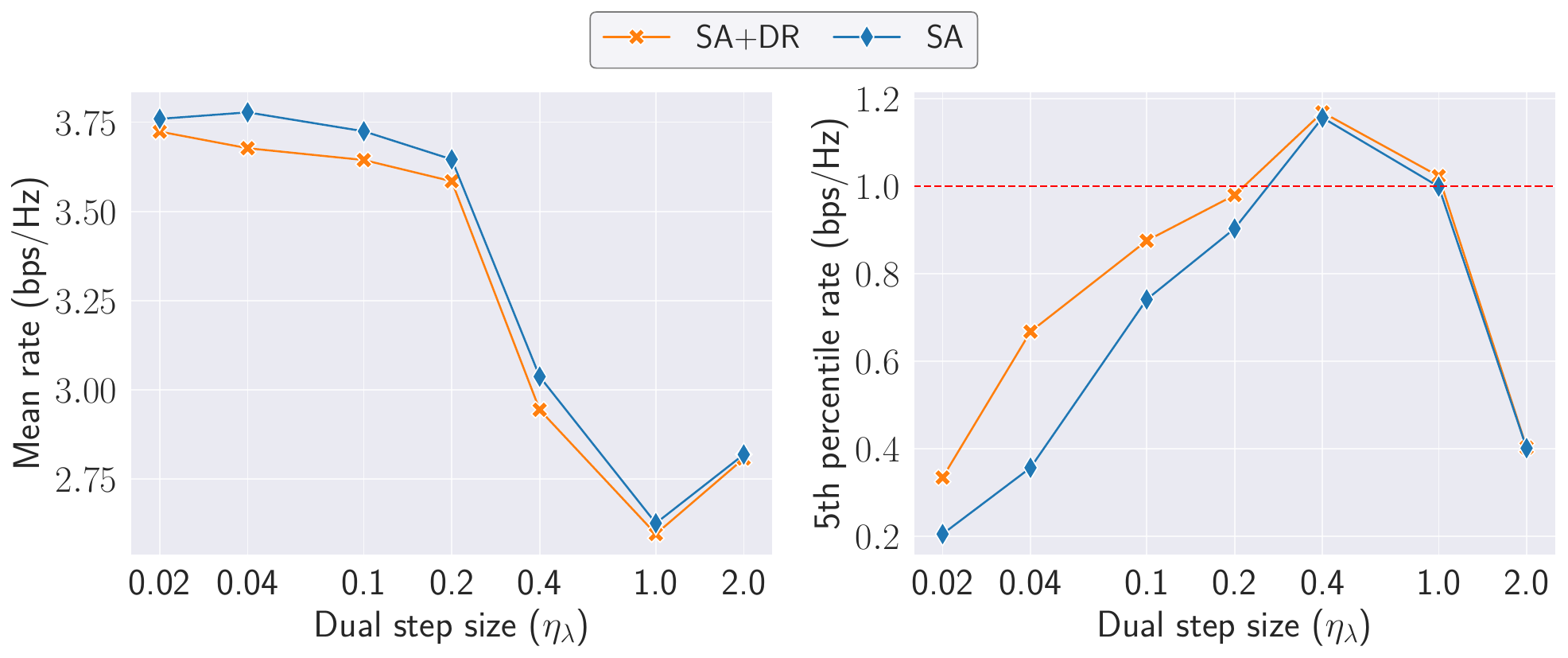}
    % \vspace{-0.5em}
        \subcaption[]{}
        % \vspace{-.5em}
    \end{minipage}
    \hfill
    \begin{minipage}{.35\linewidth}
    \vspace{1.2em}
    \includegraphics[width = \linewidth]{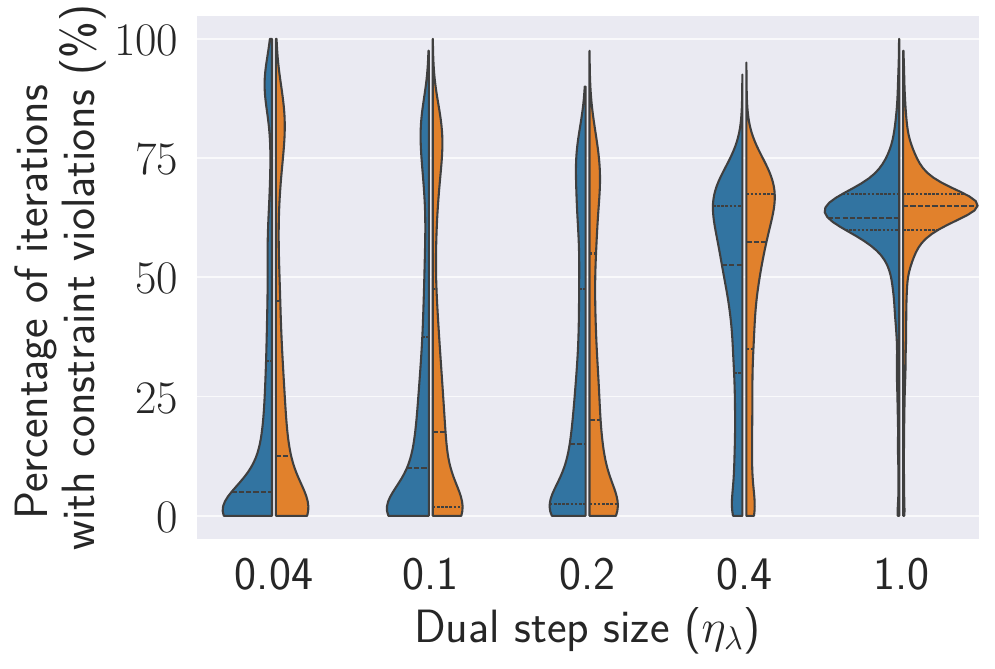}
        \subcaption[]{}
        % \vspace{-.5em}
    \end{minipage}
    \hfill
    \vfill 
    \caption{Effect of the dual step size $\eta_{\bblambda}$. (a) Mean and 5th-percentile rates. Loss of objective utility grows with $\eta_{\bblambda}$ and the benefit of near-optimal initialization steadily vanishes; while too small a step size is impractical, as percentile rates take noticeably longer to approach $f_{\min}$. Both trends match our theoretical analysis. (b) Distribution across users of the percentage of infeasible iterations, i.e., windows of $T_0$ time steps in which a user's $T_0$-averaged rate falls below $f_{\min}$. Violations generally increase with $\eta_{\bblambda}$. Overall, a sweet spot is $\eta_{\bblambda} \in [0.05, 0.25]$.
    % (a) Plot of mean and 5th percentile rates for varying dual step size choices. We see that loss of objective utility increases with growing dual step sizes and the benefit of near-optimal initialization steadily vanishes. On the other hand, too small choice of a dual step size is impractical as the time it takes for the percentile rates to approach $f_{\min}$ also noticeably increases. Note that both of these phenomenon are in line with our theoretical analysis. (b) Distribution across all users of percentage of infeasible iterations (windows of $T_0$ time steps) where a given user's $T_0$-averaged rate does not meet $f_{\min}$ for varying dual step size choices. We see that the rate of instantaneous constraint violations generally increases with the dual step size. Overall, a sweet spot for the dual step size in our setup seems to be $\eta_{\bblambda} \in [0.05, 0.25]$. 
    }
    \label{fig:dual-step-size-ablation}
\end{figure*}

% \begin{figure*}[t!]
%     \begin{minipage}{\linewidth}
%     \centering
%         \includegraphics[width=\linewidth, height=.27\linewidth]{Figures/client_10_sa_dynamics.png}
%     \end{minipage}%
%     \hfill 
%     \vfill 
%     \begin{minipage}{\linewidth}
%     \centering
%         \includegraphics[width=\linewidth, height=.27\linewidth]{Figures/client_72_sa_dynamics.png}
%     \end{minipage}%
%     \hfill 
%     \vfill 
%     \begin{minipage}{\linewidth}
%     \centering
%         \includegraphics[width=\linewidth, height=.27\linewidth]{Figures/client_71_sa_dynamics.png}
%     \end{minipage}%
%     \vfill 
%     \caption{Additional examples of SA+DR and SA algorithm roll-outs for several test network users.}
%     \label{fig:example-test-network-b}
% \end{figure*}

% \input{sections/additional_results}

\section{Proof of Lemma~\ref{lemma:boundedness_of_D}}
\begin{proof}\label{proof:boundedness_of_D}
    For any $\bblambda \in \ccalG_{\epsilon}$, the result follows from the definitions of $\ccalG_{\epsilon}$, the dual function, and strict feasibility as
    % \begin{equation}
        % \begin{aligned}
        \begin{align}
            &D^\star_{\bbtheta} + \epsilon \geq g_{\bbtheta}(\bblambda) = \max_{\bbtheta} \ccalL(\bbtheta, \bblambda) \geq \ccalL(\bbtheta^\dagger, \bblambda) \\
            &\geq f_0 \left( \h, \bbp_{\bbtheta}\big(\h; \bbtheta^\dagger  \big)  \right) + \bblambda^\top \bbf \left( \h, \bbp_{\bbtheta}\big( \h; \bbtheta^\dagger \big) \right) \\
            &\geq f_0 \left( \h, \bbp_{\bbtheta}\big(\h; \bbtheta^\dagger  \big)  \right) + \xi \| \bblambda \|_1.
            \label{eq:1_norm_bound}
        \end{align}
        % \end{aligned}
    % \end{equation}
    Rewriting~\eqref{eq:1_norm_bound}, we have
    \begin{equation}
        \norm{\bblambda}_1 \leq \frac{D^\star_{\bbtheta} + \epsilon - f_0 \left( \bbf \big( \h, \bbp_{\bbtheta}\big(\h; \bbtheta^\dagger \big) \big) \right)}{\xi},
    \end{equation}
    and the proof is complete.
\end{proof}

\section{Proof of Lemma~\ref{lemma:DGA_expected_lambda_norm_squared}}

\begin{proof}\label{proof:DGA_expected_lambda_norm_squared}
We can relate the expected squared distance of a dual iterate from an arbitrary optimal dual multiplier $\bblambda^\star \in \bbLambda^\star$ at iteration $k+1$ to the squared distance at iteration $k$ as
% \begin{subequations}
    \begin{align}
        &\E \left[ \norm{\bblambda_{k+1} - \bblambda^\star}^2 \cond \bblambda_k \right] \\
        &\leq\E \left[ \norm{\bblambda_k - \eta_{\bblambda} \widehat{\bbs}(\bblambda_k) \!-\! \bblambda^\star}^2 \cond \bblambda_k \right] \label{eq:expected_lambda_norm_squared_step_a} \\
        &\leq \norm{\bblambda_k \!-\! \bblambda^\star}^2 + \eta_{\bblambda}^2 S^2 \!-\! 2\eta_{\bblambda} \E \left[ \widehat{\bbs}(\bblambda_k)^T \left( \bblambda_k \!-\! \bblambda^\star \right) \!\cond\! \bblambda_k \right]\!\!, \label{eq:expected_lambda_norm_squared_step_b}
        \\
        &\leq \norm{\bblambda_k - \bblambda^\star}^2 +\eta_{\bblambda}^2 S^2 - 2\eta_{\bblambda} \left[ g_{\bbtheta}(\bblambda_k) - D^\star_{\bbtheta} \right] \label{eq:expected_lambda_norm_squared_step_c}
\end{align}\label{eq:expected_lambda_norm_squared}
% \end{subequations}
\noindent where~\eqref{eq:expected_lambda_norm_squared_step_a} follows by the nonexpansiveness of the projection and norm,~\eqref{eq:expected_lambda_norm_squared_step_b} follows by the boundedness of second moments and~\eqref{eq:expected_lambda_norm_squared_step_c} follows by the definition of a subgradient of the dual function. This completes the proof.   
\end{proof}

\section{Proof of Proposition~\ref{prop:convergence_dual_multiplier_averages}}
We start by proving a lemma and its immediate corollary.

\begin{lemma}[Finite-time ergodic complementary slackness]\label{lemma:finite_time_ergodic_complementary_slackness}  
For any $\bblambda^\star$, it holds that
    \begin{align}\label{eq:finite_time_ergodic_complementary_slackness}
        \E \left[ \frac{1}{K} \sum_{k=0}^{K-1} \widehat{\bbs}(\bblambda_k)^T (- \bblambda_k - \bblambda^\star) \cond \bblambda_0 \right] \leq \eta_{\bblambda} \frac{S^2}{2} + \frac{\norm{\bblambda_0 - \bblambda^\star}^2}{2\eta_{\bblambda}K}.
    \end{align}
\end{lemma}

\begin{proof}\label{proof:finite_time_ergodic_complementary_slackness}
The proof follows the steps of~\cite[Lemma~1]{calvo2021state}. By Lemma~\ref{lemma:DGA_expected_lambda_norm_squared}, we have
% \begin{subequations}
    \begin{align}
        \E& \left[ \norm{\bblambda_{k+1} - \bblambda^\star}^2 \cond \bblambda_k \right] \\
        &\hspace{-.5cm}\leq \norm{\bblambda_k \!-\! \bblambda^\star}^2 \!+\! \eta_{\bblambda}^2 S^2 - 2\eta_{\bblambda} \E \left[ \widehat{\bbs}(\bblambda_k)^T \left( \bblambda_k \!-\! \bblambda^\star \right) \cond \bblambda_k \right]. \label{eq:eq}
\end{align}
% \end{subequations}

Applying~\eqref{eq:eq} recursively and using the tower property of conditional expectation, we obtain 
\begin{align}
    \E &\left[ \frac{1}{K} \sum_{k=0}^{K-1} \widehat{\bbs}(\bblambda_k) (\bblambda_k - \bblambda^\star) \cond \bblambda_0 \right] \leq \eta_{\bblambda} \frac{S^2}{2}  \\
    &+ \frac{1}{2\eta_{\bblambda} K} \left( \norm{\bblambda_0 - \bblambda^\star}^2 - \E \big[\norm{\bblambda_K - \bblambda^\star}^2 \cond \bblambda_0 \big] \right) \nonumber \\
    &\leq \eta_{\bblambda} \frac{S^2}{2} + \frac{\norm{\bblambda_0 - \bblambda^\star}^2}{2\eta_{\bblambda} K}.
\end{align}
This completes the proof of the lemma.
\end{proof}

\begin{corollary}[Ergodic complementary slackness]\label{corollary:ergodic_complementary_slackness}
    \begin{align}
        \liminf_{K \to \infty} \E \left[ \frac{1}{K} \sum_{k=0}^{K-1} \widehat{\bbs}(\bblambda_k)^T (\bblambda_k - \bblambda^\star) \right] \leq \eta_{\bblambda} \frac{S^2}{2}.
    \end{align}
\end{corollary}
\begin{proof}\label{proof:ergodic_complementary_slackness}
Taking the liminf of both sides in~\eqref{eq:finite_time_ergodic_complementary_slackness} and marginalizing over $\bblambda_0$ proves the corollary.
\end{proof}

% \hspace{1em}
\begin{proof}\label{proof:convergence_dual_multiplier_averages}
By Lemma~\ref{lemma:finite_time_ergodic_complementary_slackness} and the convexity of the dual function, it immediately follows that
% \begin{equation}
%     \begin{aligned}
    \begin{align}
        \eta_{\bblambda} \frac{S^2}{2} &+ \frac{\norm{\bblambda_0 - \bblambda^\star}^2}{2\eta_{\bblambda}K}
        \geq \E \left[ \frac{1}{K} \sum_{k=0}^{K-1} \widehat{\bbs}(\bblambda_k)^T (\bblambda_k - \bblambda^\star) \cond \bblambda_0 \right] \\
        &\geq \E \left[ \frac{1}{K} \sum_{k=0}^{K-1} \big[ g_{\bbtheta}(\bblambda_k) - D^\star_{\bbtheta} \big] \cond \bblambda_0 \right] \\
        &\geq g_{\bbtheta} \Big( \E \big[ \bar{\bblambda}_K \cond \bblambda_0 \big] \Big) - D^\star_{\bbtheta}. \label{eq:dual_function_of_average_dual_multiplier}
    \end{align}
%     \end{aligned}
% \end{equation}

From~\eqref{eq:dual_function_of_average_dual_multiplier}, it follows that
\begin{align}
    \norm{\E \big[ \bar{\bblambda}_K \cond \bblambda_0 \big] - \bblambda^\star} \leq B_{\bblambda} \left(1 + \frac{\norm{\bblambda_0 - \bblambda^\star}^2}{K \eta_{\bblambda}^2 S^2} \right).
\end{align}
We finish the proof with
\begin{align} 
    \liminf_{k \to \infty} \; \mathrm{dist}\big( \E[\bar{\bblambda}_K], \, \bbLambda^\star \big) \leq B_{\bblambda}.
\end{align}
\end{proof}

\hspace{1em}
\section{Proof of Proposition~\ref{proposition:dual_function_best_convergence}}\label{proof:dual_function_best_convergence}
\begin{proof}

Almost sure convergence proof follows exactly the steps of \cite[Thm.~2]{ribeiro2010ergodic} and relies on a classical supermartingale convergence theorem \cite{robbins1985}. We restate the proof and refer the reader to \cite[Ch.~2]{Schirotzek1986} for a detailed convergence analysis.
% of the stochastic supergradient ascent.

\begin{lemma}{\cite{robbins1985}}\label{thm:supermartingale_convergence}
Let $\alpha_k, \beta_k$ be two nonnegative stochastic processes adapted to the filtration $\{ \ccalF_k \}_{k \geq 0}$ satisfying 
\begin{equation}
    \begin{aligned}
        \E[\alpha_{k+1} \cond \ccalF_k] \leq \alpha_k - \beta_k
    \end{aligned}
\end{equation}
for all $k \geq 0$. Then the sequence $\alpha_k$ converges almost surely and $\sum_{k=0}^{\infty} \beta_k < \infty$ almost surely.
\end{lemma}

We define two processes $\{ \alpha_k \}_{k \geq 0}$ and $\{ \beta_k \}_{k \geq 0}$,
    \begin{align}
        \alpha_k :=& \norm{\bblambda_k - \bblambda^\star}^2 \cdot \gamma_k \\
        \beta_k :=& \Big[2 \eta_{\bblambda} \left[ g_{\bbtheta}(\bblambda_k) - D^\star_{\bbtheta} \right] - \eta_{\bblambda}^2 S^2 \Big] \cdot \gamma_k,
    \end{align}
\noindent where $\gamma_k := \mathbbm{1} \left\{ \dbest{k} - D^\star_{\bbtheta} < \eta_{\bblambda} \frac{S^2}{2} \right\}$. 

Note that $\alpha_k, \beta_k, \gamma_k \geq 0$ for all $k$. Thus, it suffices to show that $\alpha_k$ and $\beta_k$ both meet the assumptions of Lemma~\ref{thm:supermartingale_convergence} and hence $\beta_k$ is almost surely summable, which yields the desired convergence result for the best dual function optimality gap.

First, we trivially have $\alpha_{k+l} = \gamma_{k+l} = \beta_{k+l} = 0$ for all $l \geq 0$ when $\alpha_k = 0$ and in particular, $\E[\alpha_{k+l} \cond \bblambda_k, \alpha_k = 0] = 0 = \alpha_k - \beta_k$. Given $\alpha_k > 0$ for a fixed $k$, we have $\gamma_k = 1$ and thus
    \begin{align}
        &\E[\alpha_{k+1} \cond \bblambda_k, \alpha_k > 0] \leq \E[\norm{\bblambda_{k+1} - \bblambda^\star}^2 \cond \bblambda_k] \\
        &\leq \norm{\bblambda_{k+1} - \bblambda^\star}^2 + \eta_{\bblambda}^2  S^2 - 2 \eta_{\bblambda} \big[ \dbest{k} - D^\star_{\bbtheta} \big] \\
        &= \alpha_k - \beta_k.
    \end{align}

% We trivially have $\alpha_{k+l} = \gamma_{k+l} = \beta_{k+l} = 0$ for all $l \geq 0$ when $\alpha_k = 0$ and in particular, $\E[\alpha_{k+l} \cond \bblambda_k, \alpha_k = 0] = 0 = \alpha_k - \beta_k$.
Combining the two cases, for all $k \geq 0$, we have
% \begin{equation}
%     \begin{aligned}
    \begin{align}
        \E[\alpha_{k+1} \cond \ccalF_k] &= \E[\alpha_{k+1} \cond \bblambda_k, \alpha_k = 0] \cdot \Pr{\alpha_k = 0} \\
        &+ \E[\alpha_{k+1} \cond \bblambda_k, \alpha_k > 0] \cdot \Pr{\alpha_k > 0} \\
        &= \alpha_k - \beta_k.
    \end{align}
%     \end{aligned}
% \end{equation}
% \end{proof}

The proof of the proposition follows by noting that $\sum_{k=0}^{\infty} \beta_k < \infty$ a.s. implies
\begin{equation}
    \begin{aligned}
    \liminf_{k \to \infty}&  ~\Big[ 2\eta_{\bblambda} [ g_{\bbtheta}(\bblambda_k) - D^\star_{\bbtheta}] - \eta_{\bblambda}^2 S^2 \Big] \\
    &\times \mathbbm{1}\left\{ \dbest{k} - D^\star_{\bbtheta} < \eta \frac{S^2}{2} \right\} = 0 \quad \text{a.s.}
    \end{aligned}    
\end{equation}
Equivalently, either $\dbest{k} - D^\star_{\bbtheta} \leq \eta_{\bblambda} \frac{S^2}{2}$ for some finite $k$ or $\liminf_{k \to \infty}{ g_{\bbtheta}(\bblambda_k) - D^\star_{\bbtheta}} = \eta_{\bblambda} \frac{S^2}{2}$ almost surely.
\end{proof}

\section{Proof of Theorem~\ref{thm:excursion_bound}}\label{proof:excursion_bound}

\begin{proof}

We first derive the exponential probability bound over $\ccalE$, and combine the exponential bound with a linear bound to extend the statement of the theorem to arbitrary deviations. 

Given $\bblambda_k \in \ccalE$ and the dual dynamics in~\eqref{eq:dual_dynamics}, we obtain
\begin{align}
    &\E \Big[g_{\bbtheta}(\bblambda_{k+1}) - g_{\bbtheta}(\bblambda_{k}) \cond \bblambda_k \Big] \nonumber \\
    &\leq \E \Big[ s_{\bbtheta}(\bblambda_{k})^\top (\bblambda_{k+1} - \bblambda_{k}) + \frac{L}{2} \| \bblambda_{k+1} - \bblambda_{k} \|^2  \cond \bblambda_{k} \Big] \label{step:lipschitz-smoothness} \\
    &= \E \left[ \widehat{s}_{\bbtheta}(\bblambda_{k})^\top \big(\!-\eta_{\bblambda} \widehat{s}_{\bbtheta}(\bblambda_{k}) \big) \!+\! \frac{L}{2} \| \bblambda_{k+1} - \bblambda_k \|^2  \cond \bblambda_k \right] \\
    &\leq -\eta_{\bblambda} \|  s_{\bbtheta}(\bblambda_{k}) \|^2 + \frac{L}{2} \eta_{\bblambda}^2 \E \left[ \| \widehat{s}_{\bbtheta}(\bblambda_{k}) \|^2 \right],
\end{align}
where we used the Lipschitz smoothness of $g_{\bbtheta}$ in~\eqref{step:lipschitz-smoothness}. By the strong convexity assumption, the squared norm of the subgradients can be lower bounded as $\|  s_{\bbtheta}(\bblambda) \|^2 \geq 2m \cdot \left[ g_{\bbtheta}(\bblambda) - D^\star_{\bbtheta}\right]$ throughout $\ccalE$. In particular, if $G_k = g_{\bbtheta}(\bblambda_k) - D^\star_{\bbtheta} = (1 + \rho)\frac{\kappa}{2} . \eta_{\bblambda} \frac{S^2}{2} $ for any $\rho > 0$, where $\kappa := L/m$ is the condition number, then we have a supermartingale relationship given by
\begin{align} \label{eq:G_k-supermartingale}
    \E[G_{k+1} \cond G_k] \leq G_k - \rho \eta_{\bblambda}^2 L \frac{S^2}{2}.
\end{align}

The sequence $\{G_k\}_{k \geq 0}$ has strictly positive expected decrements since $\E[G_{k+1} \vert G_k] \leq G_k - \mu$, for $\mu = \rho \eta_{\bblambda}^2 L \frac{S^2}{2} > 0$ by ~\eqref{eq:G_k-supermartingale}. Moreover, the sequence has almost surely bounded increments, i.e., $\Pr{ \vert G_{k+1} - G_k \vert \leq \zeta \,\vert\, G_k } = 1$ where
\begin{align}
    \zeta &= \eta_{\bblambda} S^2 \geq \sup_{k > 0} \; \eta_{\bblambda} \|  s_{\bbtheta}(\bblambda_{k}) \|  . \|  \widehat{s}_{\bbtheta}(\bblambda_{k}) \|  \\
    &\geq \sup_{k > 0} \vert s_{\bbtheta}(\bblambda_k)^\top (\bblambda_{k} - \bblambda_{k+1}) \vert\  \geq \sup_{k > 0} \vert G_{k+1} - G_k \vert.
\end{align}

Next, we leverage a result stating that if a supermartingale sequence $\{X_k \}_{k \geq 0}$ exhibits strictly positive expected decrements and almost surely bounded increments properties, then an exponential transformation of that sequence given by $Y_k := e^{\beta X_k}$ still satisfies the supermartingale property where
\begin{align}\label{eq:exponential_transform_beta}
    \beta &:= \frac{2\mu}{\mu^2 + \zeta^2} 
    = \frac{\rho L}{S^2  \left( 1 + \rho^2 \eta_{\bblambda}^2 \frac{L^2}{4} \right) }.
\end{align}
To be more precise, the transformed sequence $H_k := e^{\beta G_k}$, $k = 0, 1, \ldots$, is a supermartingale sequence with
\begin{align}
 &\mathbb{P} \big[ \vert H_{k+1} - H_k \vert \leq \beta \zeta \, \vert \, H_k \big] = 1, \\
 &\E[H_{k+1} \, \vert \, H_k] \leq H_k - \beta \mu.
\end{align}
We refer the reader to~\cite[Lemma 2]{eksin2012distributed} for the proof of the preceding step as the computations involved are not insightful.

Now, we consider an excursion $\{ G_0, G_1, \ldots, G_L \}$ starting at $G_0  = (1 + \rho) \frac{\kappa}{2} .  \eta_{\bblambda} \frac{S^2}{2}$, and finishing at $G_L$ where $L := \min \{ l > 0 \cond G_l < G_0 \}$. With $\beta$ as given in~\eqref{eq:exponential_transform_beta}, we define a new process under the exponential transformation given by $H_l = e^{\beta G_l}, l = 0, 1, \ldots, L$,
and a \emph{stopping time} 
\begin{align}
    L^\star := \min \{ l > 0 : H_{l} < H_0, H_l > e^{\beta \gamma} \}.
\end{align}
where $\gamma > 0$ is the outer excursion boundary.
Note that $L^\star$ is a proper stopping time and moreover, is almost surely finite by Lemma~\ref{proposition:dual_function_best_convergence}. If the first stopping criterion $H_{l} < H_0$ is met, then $G_l < G_0$ by the monotonicity of the exponential transform and the excursion terminates before exceeding $\gamma$. In this case, we have $L^\star = L$. When the second stopping criterion is met, the excursion exceeds $\gamma$ for the first time at some timestep $L^\star < L$. We note the following equivalences
\begin{align}
G^\dagger_0 > \gamma \; \Longleftrightarrow \; G_{L^\star} > \gamma \; \Longleftrightarrow \; H_{L^\star} > e^{\beta \gamma},
\end{align}
and obtain
\begin{align}\label{eq:equivalent_excursion_probability}
    \mathbb{P} \big [ G^\dagger_0 \geq \gamma \, \vert \, G_0 \big] = \mathbb{P} \big[ H_{L^\star} > e^{\beta \gamma} \, \vert \, H_0 \big].
\end{align}
We define the stopped process
\begin{align}\label{eq:stopped_process}
    H_{l \wedge L^\star} := \begin{cases}
        H_l &\quad \text{if $l < L^\star$}, \\
        H_{L^\star} &\quad \text{if $l \geq L^\star$}.
    \end{cases}
\end{align}
Since $L^\star$ is a.s. finite, the stopped process is a martingale by an appeal to Doob's optional stopping theorem~\cite{doob1953stochastic}. In fact, the stopped process $H_{l \wedge L^\star}$ inherits the supermartingale property of $H_l$. When $l < L^\star$, the stopped sequence~\eqref{eq:stopped_process} is identical to the transformed sequence and for $l \geq L^\star$, all values of the stopped process are equal to $H_{L^\star}$. Therefore,
\begin{align}\label{eq:stopped_process_supermartingale}
\E \big[ H_{(l + 1) \wedge L^\star} \, \vert \, H_{l \wedge L^\star} \big] &\leq H_{l \wedge L^\star}, \quad l = 0, 1, \ldots,
\end{align}
and in particular, applying~\eqref{eq:stopped_process_supermartingale} recursively yields
\begin{align}\label{eq:stopped_process_supermartingale_particular}
    \E[H_{L^\star} \, \vert \, H_0] \leq H_0.
\end{align}
Finally, we combine~\eqref{eq:equivalent_excursion_probability} and~\eqref{eq:stopped_process_supermartingale_particular} through Markov's inequality to obtain the exponential probability bound of
\begin{align}
    \mathbb{P} \big[ G^\dagger_0 \geq \gamma \, \vert \, G_0 \big]  &= \mathbb{P}[H_{L^\star} > e^{\beta \gamma} \, \vert \, H_0] \\
    &\leq \E[H_{L^\star} \, \vert \, H_0] e^{-\beta \gamma}  \\
    &\leq H_0 e^{-\beta \gamma} = \exp \left( -\beta(\gamma - G_0) \right).
\end{align}

We finish by combining the exponential probability bound on the deviations within $\ccalE$ with a linear bound outside $\ccalE$.
\end{proof}

\end{appendices}

% \clearpage
% \newpage
\bibliographystyle{IEEEbib}
\bibliography{refs.bib}

\end{document}